\documentclass[a4paper,11pt,oneside,reqno,final]{amsart}

\usepackage[color]{showkeys}     

\definecolor{refkey}{gray}{.5}   

\definecolor{labelkey}{gray}{.5} 

\usepackage{epsfig}
\usepackage{amsmath, amsthm, amsfonts}
\usepackage{graphicx}
\usepackage{color}
\usepackage{psfrag}

\numberwithin{equation}{section}

\newcommand{\R}{{\mathbb R}}

\newcommand{\al}{\alpha}

\newcommand{\ga}{\gamma}
\newcommand{\Ga}{\Gamma}

\newcommand{\ep}{\varepsilon}

\newtheorem{theo}{{\sc \bf Theorem}}[section]

\newtheorem{prop}[theo]{{\sc \bf Proposition}}

\newenvironment{rem}{\medskip\noindent{\it Remark:\/} }{\medskip}

\begin{document}

\title[Calculation of the constant factor in the six-vertex model]{Calculation of the constant factor in the six-vertex model}

\author{Pavel Bleher}
\address{Department of Mathematical Sciences,
Indiana University-Purdue University Indianapolis,
402 N. Blackford St., Indianapolis, IN 46202, U.S.A.}
\email{bleher@math.iupui.edu}

\author{Thomas Bothner}
\address{Department of Mathematical Sciences,
Indiana University-Purdue University Indianapolis,
402 N. Blackford St., Indianapolis, IN 46202, U.S.A.}
\email{tbothner@iupui.edu}



\thanks{The first author is supported in part
by the National Science Foundation (NSF) Grants DMS-0969254 and DMS-1265172.}

\date{\today}

\begin{abstract}
We calculate explicitly the constant factor $C$ in the large $N$ asymptotics 
of the partition function $Z_N$ of the six-vertex model with domain wall boundary conditions
on the critical line between the disordered and ferroelectric phases. On the critical line the weights $a,b,c$ of
the model are parameterized by a parameter $\al>1$, as $a=\frac{\al-1}{2}$, $b=\frac{\al+1}{2}$, $c=1$.  
The asymptotics of $Z_N$ on the critical line was obtained earlier in the
paper \cite{BL2} of Bleher and Liechty:  $Z_N=CF^{N^2}G^{\sqrt{N}}N^{1/4}\big(1+O(N^{-1/2})\big)$, 
where $F$ and $G$ are given by explicit expressions, but the constant factor $C>0$ was not known. 
To calculate the constant $C$, we find, by using the Riemann-Hilbert approach, an asymptotic
behavior of $Z_N$ in the double scaling limit,
as $N$ and $\al$ tend simultaneously to $\infty$ in such a way that $\frac{N}{\al}\to t\ge 0$.
Then we apply the Toda equation for the tau-function to find a structural form for $C$, as a function of $\al$, and 
we combine the structural form of $C$ and the double scaling asymptotic behavior of $Z_N$ to
calculate $C$.  
\end{abstract}
\maketitle


\section{Introduction and summary of results}

The six-vertex model is a model in statistical mechanics stated on a square lattice in $\mathbb{Z}^2$ with $N^2$ vertices and arrows on edges. The arrows obey the {\it ice-rule}: at every vertex two arrows point in and two arrows point out. This rule allows for six possible configurations which are depicted in Figure \ref{fig1}.
\begin{figure}[tbh]
  \begin{center}
  \includegraphics[width=8cm,height=6cm]{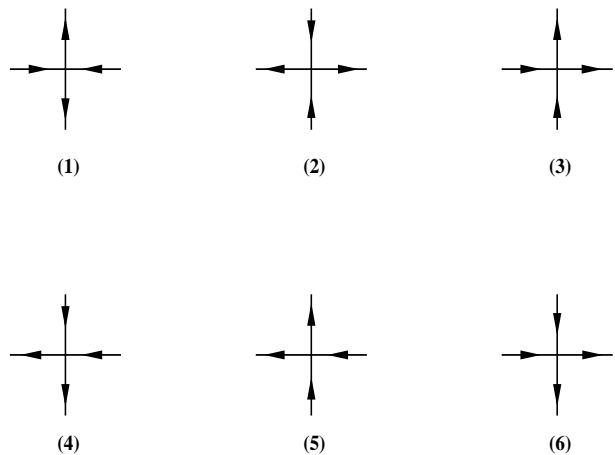}
  \end{center}
  \caption{The arrow configurations at a vertex allowed by the ice-rule}
  \label{fig1}
\end{figure}

On the lattice boundary we consider {\it domain wall boundary conditions} (DWBC), in which all arrows on the top and bottom side of the lattice point inside, and all arrows on the left and right side point outside. We depict a possible $4\times 4$ configuration with DWBC in Figure \ref{fig2} below.
\begin{figure}[tbh]
  \begin{center}
  \includegraphics[width=7cm,height=7cm]{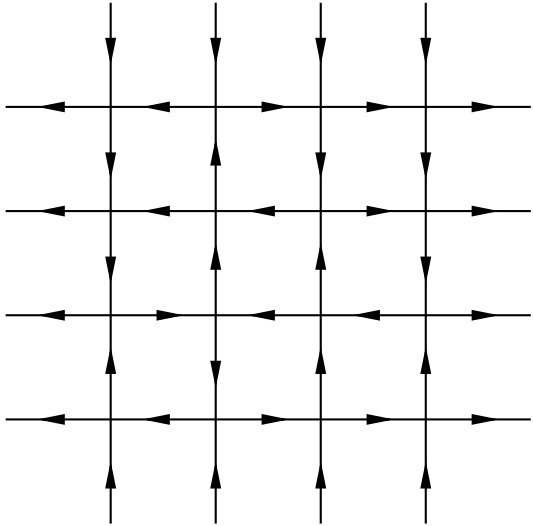}
  \end{center}
  \caption{One example of a $4\times 4$ configuration with DWBC}
  \label{fig2}
\end{figure}

\smallskip

For each vertex type $(i)$ shown in Figure \ref{fig1}, we assign a Boltzmann weight $w_i,i=1,\ldots 6$ and introduce the partition function $Z_N$ as
\begin{equation*}
	Z_N = \sum_{\textnormal{configurations}}\prod_{i=1}^6w_i^{n_i},
\end{equation*}
where $n_i$ denotes the number of vertices of type $(i)$. The partition function $Z_N$ depends, by definition, on six parameters: the weights $w_i$. Through the existence of conservation laws (compare \cite{AR,BF} or \cite{FS}) we can reduce the general case
to the case when $w_1=w_2=a$, $w_3=w=4=b$, and $w_5=w_6=c$. By using the homogeneity of the partition
function with respect to $a$, $b$, and $c$, it can be further reduced to two parameters, 
$\frac{a}{c}$ and $\frac{b}{c}$.\smallskip

The phase diagram of the model is depicted in Figure \ref{fig3}, it shows three phase regions: 
the antiferroelectric (AF) phase region, the disordered (D) phase region, and the ferroelectric (F) phase region.
\begin{figure}[tbh]
  \begin{center}
  \psfragscanon
  \psfrag{1}{\footnotesize{$0$}}
  \psfrag{2}{\footnotesize{$1$}}
  \psfrag{3}{\footnotesize{$1$}}
  \psfrag{4}{\footnotesize{$\frac{a}{c}$}}
  \psfrag{5}{\footnotesize{$\frac{b}{c}$}}
  \psfrag{6}{\footnotesize{$\textnormal{(AF)}$}}
  \psfrag{7}{\footnotesize{$\textnormal{(F)}$}}
  \psfrag{8}{\footnotesize{$\textnormal{(F)}$}}
  \psfrag{9}{\footnotesize{$\textnormal{(D)}$}}
  \includegraphics[width=6.5cm,height=6cm]{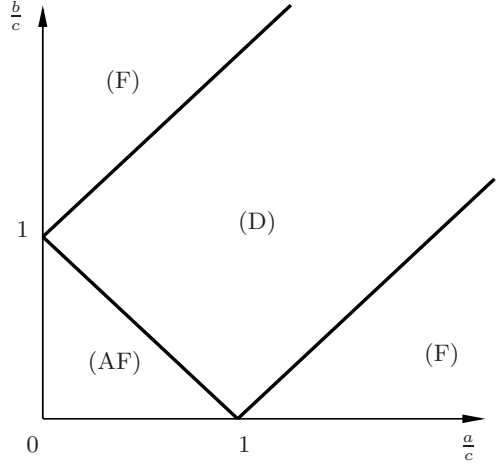}
  \end{center}
  \caption{The phase diagram of the six-vertex model, with the phases (F), (AF) and (D)}
  \label{fig3}
\end{figure}
In these phase regions the weights $a,b$ and $c$ are parameterized as follows:
\begin{align}
	(\textnormal{F})\hspace{0.75cm}&a=\sinh(t-\gamma),&b=\sinh(t+\gamma),\hspace{0.5cm}& c=\sinh(2\gamma),& 0<|\gamma|<t\label{Fpara}\\
	(\textnormal{AF})\hspace{0.75cm}&a=\sinh(\gamma-t),&b=\sinh(\gamma+t),\hspace{0.5cm}& c=\sinh(2\gamma),& 0\leq|t|<\gamma\label{AFpara}\\
	(\textnormal{D})\hspace{0.75cm}&a=\sin(\gamma-t), &b=\sin(\gamma+t),\hspace{0.5cm}& c=\sin(2\gamma),& 0\leq|t|<\gamma<\frac{\pi}{2}.\label{Dpara}
\end{align}
In the present work we calculate the asymptotics of the partition function $Z_N$ of the six-vertex model on the critical line between the disordered and ferroelectric phases which corresponds to
\begin{equation}\label{IKform3}
	\frac{b}{c}-\frac{a}{c}=1.
\end{equation}
On the critical line we use the following parameterization of the weights:
\begin{equation}\label{IKform4}
	a = \frac{\al-1}{2},\hspace{0.5cm}b=\frac{\al+1}{2},\hspace{0.5cm}c=1,\hspace{0.75cm}\al\in(1,\infty).
\end{equation}

The six-vertex model with DWBC was first studied by Korepin in \cite{K}, then further analyzed in the works \cite{I} and \cite{CIK}. This effort lead
to a determinantal formula for the partition function $Z_N$, which is commonly called the {\it Izergin-Korepin formula}.
On the D-F critical line, with weights \eqref{IKform4}, the Izergin-Korepin formula is
\begin{equation}\label{IKform1}
	Z_N(a,b,c) = \left(\frac{\al^2-1}{2}\right)^{N^2}\frac{\tau_N}{(\prod_{k=0}^{N-1}k!)^2},
\end{equation}
where $\tau_N$ is the Hankel determinant,
\begin{equation*}
	\tau_N = \det\left(\frac{d^{i+j-2}}{d\al^{i+j-2}}\varphi\right)_{i,j=1}^N,\hspace{0.5cm}\varphi(\al) = \frac{2}{\al^2-1}.
\end{equation*}
Also, as a consequence of the determinantal formula, $\tau_N$ satisfies the Toda equation,
\begin{equation}\label{toda1}
	\tau_N\tau_N''-\left(\tau_N'\right)^2=\tau_{N+1}\tau_{N-1},\hspace{0.5cm}N\geq 1,\hspace{0.5cm}(') = \frac{d}{d\al}.
\end{equation}
\smallskip

It was observed by Zinn-Justin \cite{Z} that the Hankel determinant $\tau_N$ can be connected to orthogonal polynomials: Since
\begin{equation}\label{zj0}
	\varphi(\al) = \frac{1}{\al-1}-\frac{1}{\al+1} = \int\limits_0^{\infty}e^{-\al x}\left(e^x-e^{-x}\right)\,dx
\end{equation}
we deduce the Zinn-Justin formula,
\begin{equation}\label{zj1}
	\tau_N = \frac{1}{N!}\int\limits_0^{\infty}\cdots\int\limits_0^{\infty}\prod_{j=1}^Nw(x_j)\prod_{i<j}(x_i-x_j)^2dx_1\cdots dx_N,
\end{equation}
where
\begin{equation}\label{zj2}
	w(x) = e^{-\al x}\left(e^x-e^{-x}\right).
\end{equation}
Now introduce monic orthogonal polynomials $\{p_n(x)=x^n+\ldots\}_{n\geq 0}$ with respect to the measure $d\mu(x) = w(x)dx$ 
on the half-axis $[0,\infty)$,
\begin{equation}\label{zj21}
	\int\limits_0^{\infty}p_n(x)p_m(x)d\mu(x) = h_n\delta_{nm}
\end{equation}
and obtain from \eqref{zj1} via orthogonality, that
\begin{equation}\label{zj3}
	\tau_N = \prod_{k=0}^{N-1}h_k.
\end{equation}
The latter identity allows us to rewrite the Toda equation \eqref{toda1} on the critical line in the convenient form
\begin{equation}\label{zj5}
	\left(\ln\tau_N\right)'' = \frac{h_N}{h_{N-1}},\hspace{0.5cm}(') = \frac{d}{d\al},
\end{equation}
and for the partition function, via \eqref{IKform1},
\begin{equation}\label{zj4}
	Z_N = \left(\frac{\al^2-1}{2}\right)^{N^2}\prod_{k=0}^{N-1}\frac{h_k}{(k!)^2}.
\end{equation}
In \cite{BL2}, Bleher and Liechty derive the following large $N$ asymptotics of the partition function $Z_N$:
 
\begin{theo}[see \cite{BL2}]\label{theorem1} On the critical line between the disordered and ferroelectric phase region with $\al>1$, as $N\rightarrow\infty$,
\begin{equation}\label{the20}
	Z_N = CF^{N^2}G^{\sqrt{N}}N^{1/4}\left(1+O\big(N^{-1/2}\big)\right),
\end{equation}
with
\begin{equation*}
	F = \frac{\al+1}{2},\hspace{0.75cm} G = \exp\left[-\zeta\left(\frac{3}{2}\right)\sqrt{\frac{\al-1}{2\pi}}\right].
\end{equation*}
Here $\zeta(s)$ denotes the Riemann zeta function, $C>0$, and the error term in \eqref{the20} is uniform on any compact subset of the set
\begin{equation}\label{excset1}
	\left\{\al\in\mathbb{R}:\ \al>1\right\}.
\end{equation}
\end{theo}

The main result in the present paper is an explicit evaluation of the constant factor $C>0$. We prove the following result:

\begin{theo}\label{theorem2}
The constant factor $C$ in asymptotic formula \eqref{the20} is equal to
\begin{equation}\label{the21}
	C =e^{c} (\al-1)^{1/4},
\end{equation}
where
\begin{equation}\label{the22}
	c = \frac{1}{4}\ln 2+\frac{1}{2}\ln\pi
+\frac{1}{4\pi}\sum_{n=1}^{\infty}\left[-\frac{\pi}{n}+\sum_{m=1}^{\infty}\frac{1}{(m+n)\sqrt{mn}}\right].
\end{equation}
\end{theo}
\begin{rem} Using MAPLE for the numerical evaluation of the series, we obtain that
\begin{equation*}
	\sum_{n=1}^{\infty}\left[-\frac{\pi}{n}+\sum_{m=1}^{\infty}\frac{1}{(m+n)\sqrt{mn}}\right]= -3.568781612\ldots
\end{equation*}
and 
\begin{equation*}
	c= 0.4616571210\ldots
\end{equation*}
\end{rem}

The calculation of the constant factor in the asymptotics of the partition function is a notoriously difficult problem. This problem
appears not only in exactly solvable models of statistical mechanics, such as the six-vertex model and the Ising model, 
but also  in random matrix theory, combinatorics, theory of integrable PDEs, etc.
 In different settings, the ``constant factor problem''
is studied in the works of Tracy \cite{T}, Basor and Tracy \cite{BT}, Budylin and Buslaev \cite{BB}, Ehrhardt \cite{E}, 
Deift, Its, Krasovsky,
and Zhou \cite{DIKZ}, Deift, Its, and Krasovsky \cite{DIK}, Baik, Buckingham, and DiFranco \cite{BBdF}, Bothner and Its \cite{BI}, 
Forrester \cite{F}, and others.

To prove Theorem \ref{theorem2}, we develop the Riemann-Hilbert approach to  the double scaling limit of the partition function $Z_N$,
 as both $N$ and $\al$ tend to $\infty$  in such a way that $\frac{N}{\al}\to t\ge 0$, see Theorem \ref{prop6} below. The double scaling
asymptotics of the partition function can be of interest by itself.
Then we use  the Toda equation to show that the constant $C$ can be written as 
\begin{equation}\label{the23}
	C=(\al-1)^{1/4}e^{d\al+c}.
\end{equation}
After that we apply the double scaling asymptotics of $Z_N$ to 
calculate the values of $d$ and $c$. In this way we find that $d=0$ and that $c$ is given by formula   \eqref{the22}.\smallskip

The result of Theorem \ref{theorem2} adds to the work 
of the first author with Vladimir Fokin \cite{BF}, with Karl Liechty \cite{BL1,BL2,BL3}, 
and with the second author \cite{BB}. This series of papers proves conjectures of Paul Zinn-Justin 
\cite{Z} on the large $N$ asymptotics of $Z_N$ in the phase regions. 
For the convenience of the reader, we briefly summarize obtained results and 
outline what is known about the constant factor in different phase regions.\smallskip

{\it Ferroelectric phase region}. With parameterization \eqref{Fpara}, for any $\varepsilon>0$,
\begin{equation*}
	\textnormal{(F)}\hspace{0.85cm} Z_N = CF^{N^2}G^{N}\left(1+O\left(e^{-N^{1-\varepsilon}}\right)\right),\hspace{0.5cm} N\rightarrow\infty,
\end{equation*}
where
\begin{equation*}
	C=1-e^{-4\gamma},\hspace{0.5cm} G=e^{\gamma-t},\hspace{0.5cm} F=\sinh(\gamma+t),
\end{equation*}
see \cite{BL1}, so that
 the constant factor $C$ is known explicitly in the ferroelectric phase. \smallskip

{\it Antiferroelectric phase region}. Here, with parameterization \eqref{AFpara},
\begin{equation*}
	\textnormal{(AF)}\hspace{0.75cm} Z_N = CF^{N^2}\vartheta_4(N\omega)\left(1+O\left(N^{-1}\right)\right),\hspace{0.5cm} N\rightarrow\infty,
\end{equation*}
where 
\begin{equation*}
	\omega= \frac{\pi}{2}\left(1+\frac{t}{\gamma}\right),\hspace{0.5cm} F=\frac{\pi\big(\sinh(\gamma-t)\sinh(\gamma+t)\big)\vartheta_1'(0)}{2\gamma\vartheta_1(\omega)}
\end{equation*}
and $\vartheta_1(z) = \vartheta_1(z|q),\vartheta_4(z)=\vartheta_4(z|q)$ are the Jacobi theta functions with the elliptic nome $q=e^{-\frac{\pi^2}{2\gamma}}$, see \cite{BL3}. It is known that the constant factor $C$ does not depend on $t$, so that
\begin{equation}\label{cons1}
	C=C(\ga)>0,
\end{equation}
but its exact value is not known.
\smallskip

{\it Disordered phase region}. Compare \cite{BF}, with parameterization \eqref{Dpara} for some $\varepsilon>0$,
\begin{equation*}
	\textnormal{(D)}\hspace{0.75cm} Z_N=CF^{N^2}N^{\kappa}\left(1+O\left(N^{-\varepsilon}\right)\right),\hspace{0.5cm}N\rightarrow\infty,
\end{equation*}
with
\begin{equation*}
	\kappa=\frac{1}{12}-\frac{2\gamma^2}{3\pi(\pi-2\gamma)},\hspace{0.5cm}F = \frac{\pi\big(\sin(\gamma-t)\sin(\gamma+t)\big)}{2\gamma\cos\frac{\pi t}{2\gamma}}.
\end{equation*}
It is known that the constant factor $C>0$ has the following dependence on $t$:
\begin{equation}\label{cons2}
	C = C_0(\gamma)\left(\cos\frac{\pi t}{2\gamma}\right)^{\kappa},
\end{equation}
but the function $C_0(\gamma)>0$ is not known.\smallskip

{\it Critical line between the antiferroelectric and disordered phase regions.} 
With the parameterization $a=1-x,b=1+x,c=2,|x|<1$, see \cite{BB},
\begin{equation*}
	\textnormal{(AF-D)}\hspace{0.75cm} Z_N = CF^{N^2}N^{\frac{1}{12}}\left(1+O\left(N^{-1}\right)\right),\hspace{0.5cm}N\rightarrow\infty
\end{equation*}
where
\begin{equation*}
	F = \frac{\pi(1-x^2)}{2\cos\frac{\pi x}{2}}.
\end{equation*}
The constant factor $C>0$ has the structure
\begin{equation}\label{cons3} 
	C = C_0\left(\cos\frac{\pi x}{2}\right)^{\frac{1}{12}},
\end{equation}
where the universal constant $C_0>0$ is not known.\bigskip

In the last three cases, the structural information \eqref{cons1}, \eqref{cons2}, and \eqref{cons3} 
on the constant factors is obtained by combining the results in \cite{BL3,BF,BB} with the Toda equation. 
This can also be done in the present situation (see \eqref{the22}), which then leaves us with the 
determination of the constants $c$ and $d$. In order to compute them, we use the double scaling asymptotics
of the partition function as described above.\smallskip

In \cite{BL2}, Bleher and Liechty rescaled the original weight \eqref{zj2} as
\begin{equation*}
	w^o(x) = w\left(\frac{x}{\al-1}\right) = e^{-x}-e^{-rx},\hspace{0.5cm} r=\frac{\al+1}{\al-1}>1
\end{equation*}
and studied the constants $h_n^o$ associated with the monic orthogonal polynomials $\{p_n^o(x)\}_{n\geq 0}$, satisfying the orthogonality condition
\begin{equation*}
	\int\limits_0^{\infty}p_n^o(x)p_m^o(x)w^o(x)\,dx=h_n^o\delta_{nm}.
\end{equation*}
The main technical result in the work \cite{BL2} is the following asymptotic formula for $h_N^o$: As $N\rightarrow\infty$
\begin{equation}\label{BLres1}
	\ln\left[\frac{h_N^o}{(N!)^2}\right] = -\frac{\zeta(3/2)}{2\sqrt{\pi(r-1)}N^{1/2}}+\frac{1}{4N}+O\left(N^{-3/2}\right)
\end{equation}
which holds uniformly on any compact subset of the set \eqref{excset1}. Applying \eqref{zj4}, this result implies immediately \eqref{the20}, in particular it gives the listed explicit expressions for $F$ and $G$. However we cannot take the limit $\al\rightarrow\infty$ in \eqref{BLres1}. To overcome this difficulty, in addition to weight \eqref{zj2}, we will study the related weight
\begin{equation}\label{weightb}
	w_t(x) = \frac{w(tx)}{2t} = xe^{-NV(x)},\hspace{0.5cm} V(x) = x-\frac{\tau}{t}\ln S(tx),\hspace{0.5cm}t=\frac{N}{\al} = N\tau,\hspace{0.25cm} S(x)=\frac{\sinh x}{x}
\end{equation}
and its associated monic orthogonal polynomials $\{p_{n,t}(x)\}_{n\geq 0}$, satisfying the orthogonality condition
\begin{equation*}
	\int\limits_0^{\infty}p_{n,t}(x)p_{m,t}(x)w_t(x)\,dx = h_{n,t}\delta_{nm}.
\end{equation*}
Noticing that $h_N^o=(\al-1)^{2N+1}h_N$, we will prove the following generalization of \eqref{BLres1}:
\begin{theo}\label{theo1} As $N\rightarrow\infty$,
\begin{equation}\label{the11}
	\frac{h_N}{(N!)^2} = \frac{N}{8}\tau^{2N+2}b^2\exp\left[N(l+2)+\frac{v}{N}-\frac{1}{6N}+\varepsilon_N(\tau)\right],\hspace{0.5cm}\tau=\frac{1}{\alpha}
\end{equation}
where $\varepsilon_N$ is smooth in the parameter $\tau$ with
\begin{equation*}
	|\varepsilon_N|\leq \frac{c}{(N+1)^2},\hspace{0.5cm}\left|\frac{d\varepsilon _N}{d\tau}\right|\leq \frac{c}{(N+1)^2},\hspace{0.5cm}c>0
\end{equation*}
and the stated expansion \eqref{the11} is uniform with respect to the parameter $0\leq\tau\leq 1-\varepsilon <1$. Here, the parameter $b$ is determined implicitly via the equation
\begin{equation}\label{the12}
	b = \frac{4}{1-\tau}+\frac{2\tau}{(1-\tau)t}I(2bt),\hspace{0.5cm} I(z)= -1+\frac{z}{\pi}\int\limits_0^1\sqrt{\frac{u}{1-u}}\frac{du}{e^{zu}-1},\hspace{0.25cm}z\geq 0,
\end{equation}
the parameter $l$ equals
\begin{equation}\label{the13}
	l = 4(1-\ln 2)-\frac{b}{2}(1-\tau)-b+2\ln b-\frac{2\tau}{t}(1-\ln 2)+\frac{2\tau}{t}J(2bt)+\frac{\tau}{t}\ln S(tb)
\end{equation}
with
\begin{equation}\label{the131}
	J(z) = \frac{z}{\pi}\int\limits_0^1\left(\sqrt{\frac{u}{1-u}}-\arctan\sqrt{\frac{u}{1-u}}\right)\frac{du}{e^{zu}-1},\hspace{0.5cm}z\geq 0
\end{equation}
and
\begin{equation}\label{the14}
	v = \frac{3}{4bq(0)}-\frac{q'(b)}{4q^2(b)}+\frac{47}{12bq(b)},\hspace{0.5cm}q(z) = \frac{1}{2\pi}\oint\limits_{\Gamma}
	\sqrt{\frac{w}{w-b}}\frac{V'(w)}{w-z}\,dw,
\end{equation}
where $\Gamma$ is a counter-clockwise oriented contour containing $[0,b]\cup\{z\}$ in its interior such that $V(z)$ is analytic in the interior of $\Gamma$.
\end{theo}

The proof of Theorem \ref{theo1} relies on the Riemann-Hilbert approach to the orthogonal polynomials associated with the weight $w_t$, 
\begin{equation*}
	w_t(x) = xe^{-NV(x)},\ x\geq 0;\hspace{0.85cm}V(z) = z\big(1-\tau H(tz)\big),\hspace{0.5cm} H(z) = \frac{1}{z}\ln\left(\frac{\sinh z}{z}\right)
\end{equation*}
where the potential $V$ is analytic in the strip
\begin{equation*}
	\Delta_t = \left\{z\in\mathbb{C}:\ -\frac{\pi}{t}<\textnormal{Im}\,z<\frac{\pi}{t}\right\}.
\end{equation*}
This potential is somewhat close to the class of Laguerre potentials which were considered by Vanlessen in \cite{V}. We use a combination of techniques from \cite{V} and \cite{BL2} to derive \eqref{the11}. The explicit form of $C$ in \eqref{the21} will then follow from an interplay of the Toda equation \eqref{zj5} with \eqref{the11} which we combine with \eqref{zj4}. In particular we use the fact, that, as $\al\rightarrow\infty$ and $N$ is fixed, $w_t$ has the nontrivial limit
\begin{equation*}
	\lim_{\al\rightarrow\infty}w_t(x) = xe^{-Nx},\hspace{0.5cm}x\geq 0
\end{equation*}
and therefore the limiting orthogonal polynomials are classical Laguerre polynomials.\smallskip

The structure of the rest of the article is as follows. We derive Theorem \ref{theo1} through
the Riemann-Hilbert approach to orthogonal polynomials. This requires the construction of the equilibrium measure as stated in Section \ref{eqm}, i.e. in particular,
evaluation of the right endpoint $b$ of its support, 
its density, and the Lagrange multiplier $l$. Then, following the Deift-Zhou nonlinear steepest descent roadmap, 
we carry out in Section \ref{RHa} a sequence of transformations. These will allow us to approximate the solution 
of the initial Riemann-Hilbert problem by explicit parametrices and by an iterative solution of a singular integral equation.
As an application of this analysis, we prove in Section \ref{proof_theo1} Theorem \ref{theo1}. 
In Section \ref{toda_eq}, we use the Toda equation to derive formula \eqref{the23} for  the constant $C$. 
Then, in Section \ref{dsl}, we prove  the double scaling asymptotics of the partition function.
Finally, in Section \ref{proof_theo2}, we prove Theorem \ref{theorem2}.


\section{Equilibrium measure}\label{eqm}

\subsection{Definition of the equilibrium measure and evaluation of the endpoint of its support}
We have rescaled the original weight function $w(x)$ from \eqref{zj2} as
\begin{equation*}
	w_t(x) = \frac{w(tx)}{2t} = xe^{-NV(x)},\hspace{0.5cm} V(z) = z\big(1-\tau H(tz)\big),\hspace{0.25cm} H(z) = \frac{1}{z}\ln\left(\frac{\sinh z}{z}\right),\hspace{0.25cm} t= N\tau = \frac{N}{\al},
\end{equation*} 
hence the associated monic orthogonal polynomials $\{p_{n,t}(x)\}_{n\geq 0}$ are related to the initial ones in \eqref{zj21} via the relations
\begin{equation*}
	h_n = 2t^{2n+2}h_{n,t},\hspace{0.5cm} p_n(x) = t^np_{n,t}\left(\frac{x}{t}\right),\hspace{0.25cm}n\geq 0.
\end{equation*}
Thus, after computing the large $N$ asymptotics of $h_{N,t}$, we can evaluate $\tau_N$ from \eqref{zj3} via $h_N=2t^{2N+2}h_{N,t}$.\smallskip

Notice that we can write the Hankel determinant as
\begin{equation*}
	\tau_N = \frac{2^Nt^{N(N+1)}}{N!}\int\limits_0^{\infty}\cdots\int\limits_0^{\infty}\left(\prod_{j=1}^Ny_j\right)e^{-N\sum_{j=1}^NV(y_j)}\prod_{i<j}\left(y_i-y_j\right)^2dy_1\cdots dy_N
\end{equation*}
which allows us, with the help of the empirical measure $\nu$ on $[0,\infty)$
\begin{equation*}
	\nu(s) = \frac{1}{N}\sum_{k=1}^N\delta(s-y_k),\hspace{0.5cm}\int\limits_0^{\infty}d\nu(s) = 1,
\end{equation*}
to express parts of the integrand as
\begin{equation*}
	e^{-N\sum_{j=1}^NV(y_j)}\prod_{i<j}\left(y_i-y_j\right)^2 = e^{-N^2H(\nu)},
\end{equation*}
with the energy functional
\begin{equation*}
	H(\nu) = \iint\ln|t-s|^{-1}d\nu(t)d\nu(s)+\int V(s)d\nu(s).
\end{equation*}
This observation leads to the expectation that the value of $\tau_N$, as $N\rightarrow\infty$, will be concentrated in a vicinity of the global minimum of the functional $H(\nu)$, with $\nu$ varying over
\begin{equation*}
	\mathcal{M}^1[0,\infty) = \left\{\mu\in\textnormal{Borel measures on}\ [0,\infty):\ \int d\mu=1\right\}.
\end{equation*}
But it is well known (cf. \cite{DKM,DKMVZ}) that the minimization problem
\begin{equation*}
	E^V = \inf_{\mu\in\mathcal{M}^1[0,\infty)}\left[\iint \ln|t-s|^{-1}d\mu(t)d\mu(s)+\int V(s)d\mu(s)\right]
\end{equation*}
has a unique solution $\mu=\mu^V\in\mathcal{M}^1[0,\infty)$, the {\it equilibrium measure}. We now begin to gather various characteristica of the equilibrium measure:\bigskip

As the potential $V(z)=z(1-\tau H(tz))$ is convex, the support of the equilibrium measure $\mu^V$ consists of a single interval
\begin{equation*}
	J=[0,b]\subset\mathbb{R}.
\end{equation*}
Our first goal is to derive an expansion for the right endpoint $b$, as $N\rightarrow\infty$ for different values of the double scaling parameter $t$. To this end use the $g$-function
\begin{equation}\label{emeq1}
	g(z) = \int\limits_J\ln(z-w)d\mu^V(w) = \int\limits_0^b\ln(z-w)\psi(w)dw,\hspace{0.5cm} z\in\mathbb{C}\backslash(-\infty,b]
\end{equation}
with the principal branch chosen in the logarithm. The equilibrium measure determines the $g$-function by definition, but on the other hand the $g$-function determines the equilibrium measure uniquely through a set of variational conditions: 
\begin{equation}\label{emeq2}
	z\in[0,\infty)\backslash J:\ g_+(z)+g_-(z)-V(z)-l\leq 0,\hspace{0.75cm}z\in J:\ g_+(z)+g_-(z)-V(z)-l=0,
\end{equation}
where $l\in\mathbb{R}$ is the {\it Lagrange multiplier}. The latter equality on the support $J$, leads to an additive Riemann-Hilbert problem for the unknown $g'(z)$ which is solved explicitly
\begin{equation*}
	g'(z) = \frac{1}{2\pi}\sqrt{\frac{z}{z-b}}\int\limits_0^b\sqrt{\frac{w}{b-w}}\frac{V'(w)}{z-w}\,dw,\hspace{0.5cm}z\in\mathbb{C}\backslash[0,b].
\end{equation*}
Comparing the large $z$-asymptotics of the last equation with the one obtained from \eqref{emeq1}, we derive the following defining equation on the right endpoint $b$
\begin{equation}\label{emeq3}
	\frac{1}{2\pi}\int\limits_0^b\sqrt{\frac{w}{b-w}}V'(w)\,dw = 1.
\end{equation}
Since
\begin{equation*}
	V'(z) = 1-\tau\left(\frac{\cosh(tz)}{\sinh(tz)}-\frac{1}{tz}\right)=1-\tau\left(1-\frac{1}{tz}+\frac{2}{e^{2tz}-1}\right)
\end{equation*}
we obtain from \eqref{emeq3} after the change of variables $w=bu$
\begin{equation*}
	b = \frac{4}{1-\tau}+\frac{2\tau}{(1-\tau)t}I(2bt),\hspace{0.5cm} I(z) = -1+\frac{z}{\pi}\int\limits_0^1\sqrt{\frac{u}{1-u}}\frac{du}{e^{zu}-1},\hspace{0.25cm}z\geq 0
\end{equation*}
which is equation \eqref{the12} in Theorem \ref{theo1}. We will solve the last equation for $b$ iteratively, before that, let us study the asymptotic behavior of $I(z)$ as $z\rightarrow 0$ and $z\rightarrow+\infty$.
\begin{prop}\label{prop1} The function $I(z)$ is analytic in the strip
\begin{equation*}
	\Delta = \left\{z\in\mathbb{C}:\ -\pi<\textnormal{Im}\,z<\pi\right\}
\end{equation*}
with
\begin{equation}\label{prop1eq1}
	I(z) = -\frac{z}{4}+\frac{z^2}{32}+O\big(z^3\big),\qquad z\rightarrow 0,
\end{equation}
and as $z\to+\infty$,
\begin{equation}\label{prop1eq2}
	I(z) = -1+\frac{\zeta(3/2)}{2\sqrt{\pi}\, z^{1/2}}+\frac{3\zeta(5/2)}{8\sqrt{\pi}\,z^{3/2}}
+O\big(z^{-5/2}\big).
\end{equation}
\end{prop}
\begin{proof} Analyticity of $I(z)$ in $\Delta$ follows immediately from the analyticity of $\frac{z}{e^z-1}$ in $\Delta$, hence we are left with the two asymptotic expansions. When $z\rightarrow 0$, we use the asymptotic formula,
\begin{equation*}
	\frac{1}{e^z-1} = \frac{1}{z}-\frac{1}{2}+\frac{z}{12}+O\left(z^2\right),\hspace{0.5cm}z\to 0
\end{equation*}
combined with the integrals,
\begin{equation*}
	\int\limits_0^1\frac{du}{\sqrt{u(1-u)}} = \pi,\hspace{0.5cm}\int\limits_0^1\sqrt{\frac{u}{1-u}}\,du = \frac{\pi}{2},\hspace{0.5cm}\int\limits_0^1u\sqrt{\frac{u}{1-u}}\,du = \frac{3\pi}{8}
\end{equation*}
to obtain
\begin{equation*}
	I(z) = -1+\frac{z}{\pi}\int\limits_0^1\sqrt{\frac{u}{1-u}}\left(\frac{1}{zu}-\frac{1}{2}+\frac{zu}{12}+O\left(z^2u^2\right)\right)\,du = -\frac{z}{4}+\frac{z^2}{32}+O\left(z^3\right),
\end{equation*}
which is \eqref{prop1eq1}. When $z\rightarrow+\infty$, we use the asymptotic formula,
\begin{equation*}
	\sqrt{\frac{u}{1-u}} = u^{1/2}+\frac{u^{3/2}}{2}+O\left(u^{5/2}\right),\hspace{0.5cm}u\rightarrow 0,
\end{equation*}
and the integrals,
\begin{equation*}
	\int\limits_0^{\infty}\frac{u^{1/2}}{e^u-1}\,du = \frac{\sqrt{\pi}}{2}\zeta\left(\frac{3}{2}\right),\hspace{0.5cm}\int\limits_0^{\infty}\frac{u^{3/2}}{e^u-1}\,du = \frac{3\sqrt{\pi}}{4}\zeta\left(\frac{5}{2}\right),
\end{equation*}
which gives \eqref{prop1eq2}.
\end{proof}
Let us now return to equation \eqref{the12}. Since $\al>1$, we have that $\tau<1$ and we will in fact assume from now on, that $\tau$ is separated from $1$, so that
\begin{equation*}
	0\leq \tau \leq 1-\varepsilon<1,
\end{equation*}
where $\varepsilon>0$ is fixed throughout the remainder of this paper. To solve \eqref{the12} for $b$, use iterations
\begin{equation}\label{iter1}
	b_j=\frac{4}{1-\tau}+\frac{2\tau }{(1-\tau)t}\,I(2b_{j-1}t),\ j\geq 1 \qquad b_0=\frac{4}{1-\tau}\,.
\end{equation}
Consider the mapping 
\begin{equation*}
	 b\stackrel{f}{\mapsto} \frac{4}{1-\tau}+\frac{2\tau }{(1-\tau)t}\,I(2bt)
\end{equation*}
which satisfies
\begin{equation*}
	\frac{df}{db}=\frac{4\tau }{(1-\tau)}\,I'(2bt).
\end{equation*}
From \eqref{prop1eq2}, after differentiation, we obtain the estimate,
\begin{equation*}
	\left|\frac{df}{db}\right|\le \frac{C\tau }{1+t^{3/2}}\,,\qquad t>0,\quad b>1,
\end{equation*}
where $C>0$ is a constant independent of $\tau$ and $t$, i.e. the mapping $f$ is contracting for small $\tau$ in a neighborhood of the point $b_0$,
and hence $f$ has a fixed point $b$ which can be obtained as a limit of the iterations 
\begin{equation*}
	b_j= f(b_{j-1}),\quad j\ge 1.
\end{equation*}
In addition, we obtain the estimate of the difference $|b_j-b|$ as
\begin{equation}\label{iter2}
	|b_j-b|\le \hat{C} \left(\frac{\tau }{1+t^{3/2}}\right)^{j+1},\hspace{0.25cm}\hat{C}>0.
\end{equation}
Back to \eqref{iter1}, we have in the first iteration 
\begin{equation*}
	b_1=\frac{4}{1-\tau}+\frac{2\tau}{(1-\tau)t}\,I\left(\frac{8t}{1-\tau}\right)
\end{equation*}
and in the second,
\begin{equation}\label{iter3}
	b_2=\frac{4}{1-\tau}+\frac{2\tau}{(1-\tau)t}\,I\left(2tb_1\right)
=\frac{4}{1-\tau}+\frac{2\tau}{(1-\tau)t}\,I\left(\frac{8t}{1-\tau}+\frac{4t\tau}{(1-\tau)t}\,I\left(\frac{8t}{1-\tau}\right)\right).
\end{equation}
Combining now \eqref{iter3} with \eqref{iter2} we obtain
\begin{prop}\label{prop2} As $N\to\infty$, the right endpoint $b$ of the equilibrium measure has the asymptotic behavior,
\begin{equation}\label{prop2eq1}
	b=\frac{4}{1-\tau}+\frac{2\tau}{(1-\tau)t}\,I\left(\frac{8t}{1-\tau}
+\frac{4\tau}{1-\tau}\,I\left(\frac{8t}{1-\tau}\right)\right)
+ O\left(\left(\frac{\tau }{1+t^{3/2}}\right)^3\right)\,,
\end{equation}
which is uniform with respect to the parameters $0\le \tau\le 1-\ep<1$ and $t\ge 0$.
\end{prop}
At this point we would like to collect some facts on the density $\psi(x),x\in[0,b]$ of the equilibrium measure as introduced in \eqref{emeq1}.


\subsection{ Evaluation of the density of the equilibrium measure}
We use \cite{DKMVZ} and \cite{V}, more precisely the identities
\begin{equation*}
	\psi(x) = \frac{1}{2\pi}\sqrt{\frac{b-x}{x}}\,q(x),\hspace{0.5cm} x\in[0,b]
\end{equation*}
with (compare \eqref{the14})
\begin{equation*}
	q(z) = \frac{1}{2\pi i}\oint\limits_{\Gamma}\sqrt{\frac{w}{w-b}}\frac{V'(w)}{w-z}\,dw
\end{equation*}
where $\Gamma$ is a counter-clockwise oriented contour containing $[0,b]\cup\{z\}$ in its interior such that $V(z)$ is analytic in the interior of $\Gamma$. Replacing $V'(w)$ by
\begin{equation*}
	V'(z) = 1-\tau k(tz),\hspace{0.5cm}k(z) = \frac{\cosh z}{\sinh z}-\frac{1}{z}
\end{equation*}
we obtain after applying residue theorem
\begin{equation}\label{denseq1}
	q(z) = 1+\tau s(z,t),\hspace{0.5cm}s(z,t) = -\frac{1}{2\pi i}\oint\limits_{\Gamma}\sqrt{\frac{w}{w-b}}\frac{k(tw)}{w-z}\,dw.
\end{equation}
The properties of the function $s(z,t)$ will be important for us. They are described as follows:
\begin{prop}\label{prop3} The function $s(z,t)$ is real analytic on the set $[0,b]\times\R$. As $t\to 0$,
\begin{equation}\label{prop3eq1}
	s(z,t)=-\frac{t}{3}\left(z+\frac{b}{2}\right) +\frac{t^3}{45}\left(z^3+\frac{b}{2}z^2+\frac{3b^2}{8}z+\frac{5b^3}{16}\right)+O\big(t^5\big),
\end{equation}
and as $t\to\infty$,
\begin{equation}\label{prop3eq2}
	s(z,t) =-1+ O\big(t^{-1/2}\big),\hspace{1cm}\frac{\partial s(z,t)}{\partial z} =O\big(t^{-1/2}\big)\,,
\end{equation}
uniformly in $z\in[0,b]$.
\end{prop}
\begin{proof} The integrand in \eqref{denseq1} is analytic with respect to $(z,t)\in[0,b]\times\mathbb{R}$, hence $s(z,t)$ is analytic as well and $s(z,t)$ is real-valued since the contour $\Gamma$ can be deformed to the interval $[0,b]$ covered twice. Let us now derive the asymptotic formulae for $s(z,t)$. When $t\rightarrow 0$, use
\begin{equation*}
	k(z) = \frac{z}{3}-\frac{z^3}{45}+O\left(z^5\right),\hspace{0.5cm}z\rightarrow 0
\end{equation*}
and obtain, as $t\rightarrow 0$,
\begin{equation*}
	s(z,t) = -\frac{1}{2\pi i}\oint\limits_{\Gamma}\sqrt{\frac{w}{w-b}}\,\frac{tw}{3}\,\frac{dw}{w-z} 
	+\frac{1}{2\pi i}\oint\limits_{\Gamma}\sqrt{\frac{w}{w-b}}\,\frac{(tw)^3}{45}\,\frac{dw}{w-z}+O\left(t^5\right),
\end{equation*}
hence equation \eqref{prop3eq1} follows from residue theorem, noticing that
\begin{eqnarray*}
	\sqrt{\frac{w}{w-b}}\,\frac{1}{w-z} &=& 1+\frac{1}{w}\left(z+\frac{b}{2}\right)+\frac{1}{w^2}\left(z^2+\frac{b}{2}z+\frac{3b^2}{8}\right)+\frac{1}{w^3}\bigg(z^3+\frac{b}{2}z^2\\
	&&+\frac{3b^2}{8}z+\frac{5b^3}{16}\bigg)+O\left(w^{-4}\right),\hspace{0.5cm}|w|\rightarrow\infty.
\end{eqnarray*}
For the expansions \eqref{prop3eq2}, we rewrite $s(z,t)$ as
\begin{equation}\label{denseq2}
	s(z,t) = -\frac{1}{2\pi i}\oint\limits_{\Gamma}\sqrt{\frac{w}{w-b}}\left(1-\frac{1}{tw}+\frac{2}{e^{2tw}-1}\right)\frac{dw}{w-z} = -1-\frac{1}{2\pi i}\oint\limits_{\Gamma}\sqrt{\frac{w}{w-b}}\frac{2}{e^{2tw}-1}\frac{dw}{w-z},
\end{equation}
where the last equality follows once more from residue theorem. After the change of variables $u=2tw$,
\begin{equation*}
	s(z,t) =-1 
	-\frac{1}{2\pi i}\oint\limits_{2t\Gamma}\sqrt{\frac{u}{u-2tb}} \,\frac{2}{e^{u}-1}\,\frac{du}{u-2tz}
\end{equation*}
and we now choose the contour $2t\Ga$ to be a long ``stadium'', so that it consists of two parallel segments,
$\{u=x\pm i,\; 0\le x\le 2tb\}$ and two semicircles of radius 1, around the points $u=0$ and $u=2tb$. With this choice 
\begin{equation*}
	\left|\sqrt{\frac{u}{u-2tb}} \,\frac{2}{e^{u}-1}\,\frac{1}{u-2tz}\right|\le \frac{C}{\sqrt{1+t}}\,,\quad u\in 2t\Ga,
\end{equation*}  
and we obtain the first estimation in \eqref{prop3eq2}. For the second, we differentiate with respect to $z$, i.e.
\begin{equation*}
	\frac{\partial s(z,t)}{\partial z}  = -\frac{1}{2\pi i}\oint\limits_{\Gamma}\sqrt{\frac{w}{w-b}}\,\frac{2}{e^{2tw}-1}\,\frac{dw}{(w-z)^2},
\end{equation*}
and change again variables $u=2tw$. Estimating the latter integral from above, we obtain the remaining estimation in \eqref{prop3eq2}.
\end{proof}
We can now combine \eqref{denseq1} with Proposition \ref{prop3} and deduce
\begin{equation}\label{denseq3}
	\sup_{0\leq y\leq b}\big|q(y)-1+\tau\big| 
=  O\left(\frac{\tau}{\sqrt{1+t}}\right),\qquad 
\sup_{0\leq y\leq b}\big|q'(y)\big|= O\left(\frac{\tau}{\sqrt{1+t}}\right),
\end{equation}
which is uniform with respect to the parameters $0\leq \tau\leq 1-\varepsilon<1$ and $t\geq 0$.\smallskip

We are left with the computation of the Lagrange multiplier.


\subsection{Evaluation of the Lagrange multiplier}
We will compute the multiplier $l$ via \eqref{emeq2},
\begin{equation*}
	l = 2g(b)-V(b),
\end{equation*}
in other words, we have to compute two quantities. For $g(b)$, use the formula
\begin{equation}\label{lageq1}
	g(b) = \ln b-\int\limits_b^{\infty}\left(\omega(z)-\frac{1}{z}\right)\,dz
\end{equation}
which involves the resolvent $\omega(z)\equiv g'(z)$ and which can be derived immediately from the expansion
\begin{equation*}
	g'(z) = \int\limits_0^b\frac{d\mu^V(w)}{z-w} = \frac{1}{z}+O\left(z^{-2}\right),\hspace{0.5cm}z\rightarrow\infty.
\end{equation*}
As we have already seen,
\begin{equation*}
	\omega(z) = \frac{1}{2\pi}\sqrt{\frac{z-b}{z}}\int\limits_0^b\sqrt{\frac{w}{b-w}}\frac{V'(w)}{z-w}\,dw,
\end{equation*}
and the latter equality can be rewritten as (cf. \cite{DKMVZ})
\begin{equation*}
	\omega(z) = \frac{V'(z)}{2}-\sqrt{\frac{z-b}{z}}\frac{q(z)}{2},\hspace{0.5cm}z\in\mathbb{C}\backslash[0,b].
\end{equation*}
Back to \eqref{lageq1}, this implies
\begin{equation*}
	g(b) = \ln b-\frac{1}{2}\int\limits_b^{\infty}\left(V'(z)-\sqrt{\frac{z-b}{z}}\,q(z)-\frac{2}{z}\right)\,dz.
\end{equation*}
Deforming the contour of integration in \eqref{denseq2} and evaluating the residue at $w=z$, we have
\begin{equation*}
	s(z,t) = -1-\sqrt{\frac{z}{z-b}}\frac{2}{e^{2tz}-1}-\frac{1}{\pi}\int\limits_0^b\sqrt{\frac{w}{b-w}}\frac{2}{e^{2tw}-1}\frac{dw}{w-z}
\end{equation*}
and with 
\begin{equation*}
	V'(z) = 1-\tau\left(1-\frac{1}{tz}+\frac{2}{e^{2tz}-1}\right)
\end{equation*}
therefore
\begin{eqnarray*}
	g(b) &=& \ln b-\frac{1}{2}\int\limits_b^{\infty}\bigg[1-\tau+\frac{\tau}{tz}-\frac{2\tau}{e^{2tz}-1}-\sqrt{\frac{z-b}{z}}\bigg(1-\tau-\sqrt{\frac{z}{z-b}}\,\frac{2\tau}{e^{2tz}-1}\\
	&&-\frac{\tau}{\pi}\int_0^b\sqrt{\frac{w}{b-w}}\,\frac{2}{e^{2tw}-1}\,\frac{dw}{w-z}\bigg)-\frac{2}{z}\bigg]dz,
\end{eqnarray*}
or after simplifications,
\begin{eqnarray}\label{lageq2}
	g(b) &=& \ln b-\frac{1}{2}\int\limits_b^{\infty}\bigg[(1-\tau)\left(1-\sqrt{\frac{z-b}{z}}\right)+\frac{\tau}{tz}-\frac{2}{z}+\sqrt{\frac{z-b}{z}}\,\frac{\tau}{\pi }\nonumber\\
	&&\times\int_0^b\sqrt{\frac{w}{b-w}}\,\frac{2}{e^{2tw}-1}\,\frac{dw}{w-z}\bigg]dz.
\end{eqnarray}
At this point we use \eqref{the12} and write
\begin{equation*}
 -\frac{\tau}{z\pi}\int\limits_0^b\sqrt{\frac{w}{b-w}}\frac{2}{e^{2tw}-1}dw = \frac{1}{z}\left(2-\frac{b}{2}\left(1-\tau\right)-\frac{\tau}{t}\right),
\end{equation*}
hence \eqref{lageq2} reads as
\begin{equation}\label{lageq3}
	g(b) = \ln b+I_1+I_2,
\end{equation}
where
\begin{equation*}
	I_1 = -\frac{1}{2}\int\limits_b^{\infty}\bigg[(1-\tau)\left(1-\sqrt{\frac{z-b}{z}}\right)
+\frac{\tau}{tz}-\frac{2}{z}+\frac{1}{z}\sqrt{\frac{z-b}{z}}\left(2-\frac{b}{2}\left(1-\tau\right)-\frac{\tau}{t}\right)\bigg]dz,
\end{equation*}
and
\begin{equation*}
	I_2 = \frac{\tau}{\pi}\int\limits_b^{\infty}\sqrt{\frac{z-b}{z}}\,
\frac{1}{z}\left[\,\int\limits_0^b\sqrt{\frac{w}{b-w}}\,\frac{w\,dw}{(e^{2tw}-1)(z-w)}\,\right]\,dz.
\end{equation*}
The term $I_1$ is calculated explicitly,
\begin{equation*}
	I_1 = 2(1-\ln 2) -\frac{b}{4}\left(1-\tau\right)-\frac{\tau}{t}\left(1-\ln 2\right)
= 2-2\ln 2 -\frac{b}{4}\left(1-\tau\right)-\frac{1}{N}+\frac{ \ln 2}{N},
\end{equation*}
and in $I_2$ we change the order of integration: Since 
\begin{equation*}
	\int\limits_b^{\infty}\sqrt{\frac{z-b}{z}}\frac{dz}{z(z-w)}
=\frac{2}{w}\left(1-\sqrt{\frac{b-w}{w}}\arctan\sqrt{\frac{w}{b-w}}\right),\quad 0<w< b,
\end{equation*}
we obtain
\begin{equation}\label{lageq4}
	I_2 = \frac{2\tau}{\pi}\int\limits_0^b\left(\sqrt{\frac{w}{b-w}}-\arctan\sqrt{\frac{w}{b-w}}\right)\frac{dw}{e^{2tw}-1} = \frac{\tau}{t}J(2bt)
\end{equation}
with (compare \eqref{the131})
\begin{equation*}
	J(z) = \frac{z}{\pi}\int\limits_0^1\left(\sqrt{\frac{u}{1-u}}-\arctan\sqrt{\frac{u}{1-u}}\right)\frac{du}{e^{zu}-1}.
\end{equation*}
Some important properties of the function $J(z)$ are summarized below.
\begin{prop}\label{prop4} The function $J(z)$ is analytic in the horizontal strip $\Delta$. As $z\rightarrow 0$,
\begin{equation}\label{prop4eq1}
	J(z) = 1-\ln 2 -\frac{z}{8}+\frac{7z^2}{384}+O\left(z^3\right),
\end{equation}
and as $z\rightarrow+\infty$,
\begin{equation}\label{prop4eq2}
 J(z)= \frac{\zeta(5/2)}{4\sqrt{\pi}z^{3/2}}+\frac{9\zeta(7/2)}{16\sqrt{\pi}z^{5/2}}+O\left(z^{-7/2}\right).
\end{equation}
\end{prop}
\begin{proof} Our reasoning is almost identical to the one given in the proof of Proposition \ref{prop1}, in particular analyticity follows from the analyticity of $\frac{z}{e^z-1}$ in $\Delta$. When $z\rightarrow 0$, we use the integrals
\begin{equation*}
	\int\limits_0^1\left(\sqrt{\frac{u}{1-u}}-\arctan\sqrt{\frac{u}{1-u}}\right)\frac{du}{u} = \pi(1-\ln 2)
\end{equation*}
\begin{equation*}
	\int\limits_0^1\left(\sqrt{\frac{u}{1-u}}-\arctan\sqrt{\frac{u}{1-u}}\right)\,du =\frac{\pi}{4},\ \ \ \
	\int\limits_0^1\left(\sqrt{\frac{u}{1-u}}-\arctan\sqrt{\frac{u}{1-u}}\right)u\,du =\frac{7\pi}{32}
\end{equation*}
and obtain
\begin{eqnarray*}
	J(z) &=& \frac{z}{\pi}\int\limits_0^1\left(\sqrt{\frac{u}{1-u}}-\arctan\sqrt{\frac{u}{1-u}}\right)\left(\frac{1}{zu}-\frac{1}{2}+\frac{zu}{12}+O\left(z^2u^2\right)\right)\,du\\
	&=&1-\ln 2-\frac{z}{8}+\frac{7z^2}{384}+O\left(z^3\right),
\end{eqnarray*}
which is \eqref{prop4eq1}. When $z\rightarrow+\infty$, we use the expansion
\begin{equation*}
	\sqrt{\frac{u}{1-u}}-\arctan\sqrt{\frac{u}{1-u}}=\frac{u^{3/2}}{3}+\frac{3}{10}u^{5/2}+O\left(u^{7/2}\right),\hspace{0.5cm}u\rightarrow 0,
\end{equation*}
and the integrals,
\begin{equation*}
	\int\limits_0^{\infty}\frac{u^{3/2}}{e^u-1}\,du = \frac{3\sqrt{\pi}}{4}\zeta\left(\frac{5}{2}\right),\hspace{0.5cm}\int\limits_0^{\infty}\frac{u^{5/2}}{e^u-1}\,du = \frac{15\sqrt{\pi}}{8}\zeta\left(\frac{7}{2}\right),
\end{equation*}
which gives \eqref{prop4eq2}.
\end{proof}
In the end we go back to \eqref{lageq3} and summarize
\begin{equation*}
	g(b) =\ln b + 2(1-\ln 2)-\frac{b}{4}(1-\tau)-\frac{\tau}{t}(1-\ln 2)+\frac{\tau}{t}J(2bt),
\end{equation*}
which, combined with $V(b)=b-\frac{\tau}{t}\ln S(bt)$, gives equation \eqref{the13}, i.e.
\begin{equation}\label{lageq5}
	l = 4(1-\ln 2)-\frac{b}{2}(1-\tau)-b+2\ln b-\frac{2\tau}{t}(1-\ln 2)+\frac{2\tau}{t}J(2bt)+\frac{\tau}{t}\ln S(bt).
\end{equation}
At this point we can begin the asymptotical analysis.


\section{Riemann-Hilbert approach}\label{RHa}

\subsection{Riemann-Hilbert characterization for orthogonal polynomials}
We will solve the Fokas-Its-Kitaev \cite{FIK} Riemann-Hilbert problem (RHP) for orthogonal polynomials asymptotically: this problem requires the construction of a $2\times 2$ piecewise analytic matrix-valued function $Y(z) = Y^{(n)}(z)$ such that
\begin{itemize}
	\item $Y^{(n)}(z)$ is analytic for $z\in\mathbb{C}\backslash[0,\infty)$
	\item If we orient the half ray $[0,\infty)$ from left to right, the limiting values of $Y^{(n)}(z)$ from either side are related via the equation 
	\begin{equation*}
		Y_+^{(n)}(z) = Y_-^{(n)}(z)\begin{pmatrix}
		1 & w_t(z)\\
		0 & 1\\
		\end{pmatrix},\hspace{0.5cm} z\in[0,\infty)
	\end{equation*}
	\item At the endpoint $z=0$, $Y(z)$ remains bounded, i.e.
	\begin{equation*}
		Y(z) = O(1),\hspace{0.5cm} z\rightarrow 0,\ \ z\in\mathbb{C}\backslash[0,\infty)
	\end{equation*}
	\item As $z\rightarrow\infty$, we have
	\begin{equation*}
		Y^{(n)}(z) = \Big(I+O\big(z^{-1}\big)\Big) z^{n\sigma_3},\hspace{0.5cm} \sigma_3 = \begin{pmatrix}
		1 & 0\\
		0 & -1\\
		\end{pmatrix}.
	\end{equation*}
\end{itemize}
The unique solution $Y^{(n)}(z)$ to the latter problem equals
\begin{equation*}
	Y^{(n)}(z) = \begin{pmatrix}
	 p_{n,t}(z) & \frac{1}{2\pi i}\int_0^{\infty}p_{n,t}(s)(s)w_t(s)\frac{ds}{s-z}\\
	 \gamma_{n-1}p_{n-1,t}(z) & \frac{\gamma_{n-1}}{2\pi i}\int_0^{\infty}p_{n-1,t}(s)w_t(s)\frac{ds}{s-z}\\
	 \end{pmatrix}
\end{equation*}
where $p_{n,t}(z)=z^n+\ldots$ is precisely the $n^{\textnormal{th}}$ monic orthogonal polynomial subject to the measure $d\mu(s) = w_t(s)ds$ supported on the half-ray $[0,\infty)$. Moreover, 
\begin{equation*}
	\gamma_n = -\frac{2\pi i}{h_{n,t}},\hspace{0.5cm} h_{n,t} = \int\limits_0^{\infty}p_{n,t}^2(s)d\mu(s),
\end{equation*}
and in addition, $Y^{(n)}(z)z^{-n\sigma_3}$ has a full asymptotic expansion near infinity:
\begin{equation*}
	Y^{(n)}(z)z^{-n\sigma_3} = I+\frac{Y_1^{(n)}}{z}+O\left(z^{-2}\right),\hspace{0.5cm}z\rightarrow\infty,\hspace{0.25cm}Y_k^{(n)} = \left(Y_k^{(n)}\right)_{ij}.
\end{equation*}
This expansion connects to the normalizing constants via
\begin{equation}\label{FIKeq1}
	h_{n,t} = -2\pi i\left(Y_1^{(n)}\right)_{12}.
\end{equation}
Since we want to compute the large $N$ asymptotics of $h_{N,t}$, we will need to solve the latter RHP for $Y(z)=Y^{(N)}(z)$. Such an asymptotic solution can be derived by applying the Deift-Zhou nonlinear steepest descent method \cite{DZ} paired with techniques which have been developed in \cite{DKMVZ,V} and \cite{BL2}. In short, we will approximate the solution $Y(z)$ with the help of solutions of certain Riemann-Hilbert model problems. The necessary steps are worked out in the subsections below.


\subsection{First transformation of the RHP - normalization}
Recall \eqref{emeq1} and introduce
\begin{equation}\label{normeq1}
	Y(z) = \exp\bigg(\frac{Nl}{2}\sigma_3\bigg)T(z)\exp\bigg(N\Big(g(z)-\frac{l}{2}\Big)\sigma_3\bigg),\ \ z\in\mathbb{C}\backslash\mathbb{R}.
\end{equation}
The analytical properties of $T(z)$ are summarized in the following:
\begin{itemize}
	\item $T(z)$ is analytic for $z\in\mathbb{C}\backslash[0,\infty)$
	\item From the jump properties of $g(z)$, compare \eqref{emeq2}, we get that
	\begin{equation}\label{Tj:1}
		T_+(z) = T_-(z)\begin{pmatrix} 
		e^{-N(g_+-g_-)} & z\\
		0 & e^{N(g_+-g_-)}\\
		\end{pmatrix},\ \ \ z\in[0,b]
	\end{equation}
	and
	\begin{equation}\label{Tj:2}
		T_+(z)=T_-(z)\begin{pmatrix}
		1 & ze^{N(g_++g_--V-l)}\\
		0 & 1\\
		\end{pmatrix},\ \ \ z\in[0,\infty)\backslash[0,b].
	\end{equation}
	\item As $z\rightarrow 0$ and $z\in\mathbb{C}\backslash[0,\infty)$, the function $T(z)$ is bounded,
	\begin{equation*}
		T(z) = O(1),\hspace{0.5cm}z\to 0,\ \ z\in\mathbb{C}\backslash[0,\infty)
	\end{equation*}
	\item At infinity, the transformed function $T(z)$ is now normalized as 
	\begin{equation*}
		T(z) = I+O\big(z^{-1}\big),\hspace{0.5cm}z\rightarrow\infty.
	\end{equation*}
\end{itemize}
Consider the jumps \eqref{Tj:1} and \eqref{Tj:2}: First, by the Euler Lagrange variational condition \eqref{emeq2},
\begin{equation*}
	g_+(z)+g_-(z)-V(z)-l <0,\hspace{0.5cm} z\in[0,\infty)\backslash[0,b],
\end{equation*}
hence for $z\in(b+\eta,\infty),\eta>0$ fixed,
\begin{equation*}
	\begin{pmatrix}
	1 & ze^{N(g_++g_--V-l)}\\
	0 & 1\\
	\end{pmatrix}\longrightarrow I,\hspace{0.5cm} N\rightarrow \infty
\end{equation*}
where the stated convergence is exponentially fast. Secondly for $[0,b]$: Since $g_-(z) = V(z)-g_+(z)+l,z\in[0,b]$ the function
\begin{equation*}
	G(z) = g_+(z)-g_-(z) = 2g_+(z)-V(z)-l
\end{equation*}
can be analytically continued in a (in general $t$-dependent) neighborhood of the line segment $[0,b]$ into the upper halfplane. Here we use in particular that $V(z)$ is analytic in the strip $\Delta_t$. But since
\begin{equation*}
	G(z) = 2\pi i\int\limits_z^b\psi(w)dw,\ \ z\in[0,b]
\end{equation*}
and therefore
\begin{equation*}
	\frac{d}{dy}G(z+iy)\Big|_{y=0} = 2\pi\psi(z)>0,\hspace{0.5cm} z\in(0,b),
\end{equation*}
we see that the stated (local) continuation of $G(z)$ into the upper half-plane satisfies
\begin{equation}\label{normeq2}
	\textnormal{Re}\ G(z) >0\hspace{0.5cm}\textnormal{for}\ \ \textnormal{Im}\ z>0. 
\end{equation}
In the lower halfplane the argument is similar:
\begin{equation*}
	G(z) = -2g_-(z)+V(z)+l
\end{equation*}
admits analytical (in general into a $t$-dependent neighborhood) continuation into the lower half-plane so that
\begin{equation}\label{normeq3}
	\textnormal{Re}\ G(z)<0\hspace{0.5cm}\textnormal{for}\ \ \textnormal{Im}\ z<0.
\end{equation}
The continuations motivate the use of the matrix factorization
\begin{eqnarray*}
	\begin{pmatrix}
	e^{-N(g_+(z)-g_-(z))} & z\\
	0 & e^{N(g_+(z)-g_-(z))}\\
	\end{pmatrix} &=&\begin{pmatrix}
	1 & 0\\
	\frac{1}{z}e^{NG(z)} & 1\\
	\end{pmatrix}\begin{pmatrix}
	0 & z\\
	-\frac{1}{z} & 0\\
	\end{pmatrix}\begin{pmatrix}
	1 & 0\\
	\frac{1}{z}e^{-NG(z)} & 1\\
	\end{pmatrix}\\
	&=& S_{L_1}S_PS_{L_2},\ \ \ z\in(0,b),
\end{eqnarray*}
and thus the second transformation of the RHP.


\subsection{Second transformation of the RHP - opening of lenses}
We let $\mathcal{L}^{\pm}$ denote the upper (lower) lens, shown in Figure \ref{fig4}. Define
\begin{equation}\label{opeq1}
	S(z) = \left\{
                                   \begin{array}{ll}
                                     T(z)S_{L_2}^{-1}, & \hbox{$z\in\mathcal{L}^+$,} \\
                                     T(z)S_{L_1}, & \hbox{$z\in\mathcal{L}^-$,} \\
                                     T(z), & \hbox{else,} 
                                   \end{array}
                                 \right.
\end{equation}
so that $S(z)$ solves the following RHP
\begin{figure}[tbh]
  \begin{center}
  \psfragscanon
  \psfrag{1}{\footnotesize{$0$}}
  \psfrag{2}{\footnotesize{$b$}}
  \psfrag{3}{\footnotesize{$\mathcal{L}^+$}}
  \psfrag{6}{\footnotesize{$\mathcal{L}^-$}}
  \psfrag{7}{\footnotesize{$\gamma^+$}}
  \psfrag{10}{\footnotesize{$\gamma^-$}}
  \includegraphics[width=10cm,height=3cm]{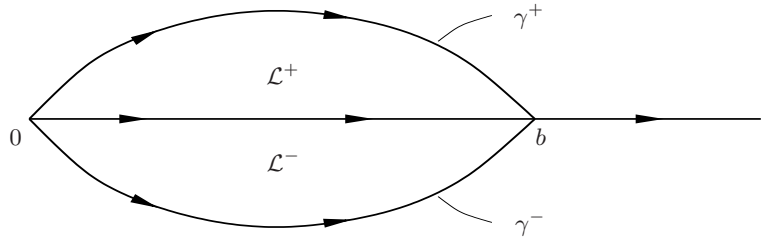}
  \end{center}
  \caption{The second transformation - opening of lenses}
  \label{fig4}
\end{figure}
\begin{itemize}
	\item $S(z)$ is analytic for $z\in\mathbb{C}\backslash([0,\infty)\cup \Gamma)$, with $\Gamma=\gamma^+\cup \gamma^-$
	\item The jumps, with orientation fixed as in Figure \ref{fig4}, are as follows
	\begin{equation*}
		S_+(z) = S_-(z)\left\{
                                   \begin{array}{ll}
                                     \bigl(\begin{smallmatrix}
                                     1 & 0\\
                                     \frac{1}{z}e^{-NG(z)} & 1\\
                                     \end{smallmatrix}\bigl), & \hbox{$z\in\gamma^+$,} \\
                                     \bigl(\begin{smallmatrix}
                                     0 & z\\
                                     -\frac{1}{z} & 0\\
                                     \end{smallmatrix}\bigr), & \hbox{$z\in(0,b)$,} \\
                                     \bigl(\begin{smallmatrix}
                                     1 & ze^{N(g_++g_--V-l)}\\
                                     0 & 1\\
                                     \end{smallmatrix}\bigr), & \hbox{$z\in[0,\infty)\backslash[0,b]$,} \\
                                     \bigl(\begin{smallmatrix}
                                     1 & 0\\
                                      \frac{1}{z}e^{NG(z)}& 1\\
                                      \end{smallmatrix}\bigr), & \hbox{$z\in\gamma^-$.}
                                   \end{array}
                                 \right.
  \end{equation*}
  \item For the behavior at the origin, we see from the behavior of $T(z)$ and \eqref{opeq1}, that
  \begin{equation}\label{opeq2}
  	S(z) = O(1),\hspace{0.5cm} z\rightarrow 0,\ z\in\mathbb{C}\backslash\big(\mathcal{L}^+\cup\mathcal{L}^-\big)
  \end{equation}
  and 
  \begin{equation}\label{opeq3}
  	S(z) = O\big(z^{-1}\big),\hspace{0.5cm} z\rightarrow 0,\ z\in\mathcal{L}^+\cup\mathcal{L}^-.
  \end{equation}
  \item As $z\rightarrow\infty$, we have $S(z)=I+O\big(z^{-1}\big)$.
\end{itemize}
Recalling \eqref{normeq2} and \eqref{normeq3} as well as the behavior of the jump matrix on the infinite ray $(b,\infty)$ we expect (and justify rigorously below) that as $N\rightarrow\infty$, $S(z)$ converges to a solution of a RHP, in which the only jump is on the line segment $(0,b)$. In more detail, this model RHP reads as follows.


\subsection{The model RHP}
Find a piecewise analytic $2\times 2$ matrix valued function $M(z)$ such that
\begin{itemize}
	\item $M(z)$ is analytic for $z\in\mathbb{C}\backslash[0,b]$
	\item Along $(0,b)$, we have the boundary relation
	\begin{equation*}
		M_+(z) = M_-(z)\begin{pmatrix}
		0 & z\\
		-\frac{1}{z} & 0\\
		\end{pmatrix},\hspace{0.5cm} z\in(0,b)
	\end{equation*}
	\item The function $M(z)$ is square integrable on $[0,b]$
	\item As $z\rightarrow\infty$, the function is normalized as 
	\begin{equation*}
		M(z)=I+O\big(z^{-1}\big)
	\end{equation*}
\end{itemize}
We compute a solution to this problem by introducing
\begin{equation*}
	N(z) = M(z)\mathcal{D}(z)^{\sigma_3},\hspace{0.5cm} z\in\mathbb{C}\backslash[0,b]
\end{equation*}
where the scalar Szeg\"o function $\mathcal{D}(z)$ satisfies
\begin{equation}\label{M:1}
	\mathcal{D}_+(z)\mathcal{D}_-(z) = z,\hspace{0.5cm}z\in[0,b].
\end{equation}
Such a function indeed exists, namely
\begin{equation*}
	\mathcal{D}(z) = \exp\Bigg[\frac{\sqrt{z(z-b)}}{2\pi i}\int\limits_0^b\frac{\ln w}{\sqrt{w(w-b)}_+}\frac{dw}{w-z}\Bigg] = \sqrt{\frac{bz}{2}}\bigg(z-\frac{b}{2}+\sqrt{z(z-b)}\bigg)^{-1/2}
\end{equation*}
where we choose principal branches for all fractional power functions. The latter choice of $\mathcal{D}(z)$ transforms the original model problem to a RHP for $N(z)$ with jump
\begin{equation*}
	N_+(z) = N_-(z)\begin{pmatrix}
	0 & 1\\
	-1 & 0\\
	\end{pmatrix},\hspace{0.5cm}z\in[0,b]
\end{equation*}
which is solved via diagonalization. Noticing further that
\begin{equation*}
	\mathcal{D}(z) = \frac{\sqrt{b}}{2}\bigg(1+\frac{b}{4z}+\frac{b^2}{8z^2}+O\big(z^{-3}\big)\bigg),\hspace{0.5cm} z\rightarrow\infty
\end{equation*}
we obtain 
\begin{equation}\label{modeq1}
	M(z) = \bigg(\frac{\sqrt{b}}{2}\bigg)^{\sigma_3}\frac{1}{2}\begin{pmatrix}
	\delta+\delta^{-1} & i(\delta-\delta^{-1})\\
	-i(\delta-\delta^{-1})& \delta+\delta^{-1}\\
	\end{pmatrix}\mathcal{D}^{-\sigma_3}(z)
\end{equation}
with 
\begin{equation*}
	\delta(z) = \bigg(\frac{z}{z-b}\bigg)^{1/4}
\end{equation*}
defined on $\mathbb{C}\backslash[0,b]$ with its branch such that $(\frac{z}{z-b})^{1/4}\rightarrow 1$ as $z\rightarrow+\infty,\ \textnormal{arg}\, z=0$. Before moving on, we note for future purposes that
\begin{equation}\label{modeq2}
	M(z) = I+\frac{b}{4z}\begin{pmatrix}
	-1& \frac{ib}{4}\\
	\frac{4}{ib} & 1\\
	\end{pmatrix}+O\big(z^{-2}\big),\hspace{0.5cm}z\rightarrow\infty.
\end{equation}


\subsection{Construction of a parametrix at $z=b$}
For a small neighborhood $\mathcal{U}$ of the point $b$, observe that
\begin{equation}\label{pl:1}
	G(z) = 2g(z)-V(z)-l= \int\limits_z^b\sqrt{\frac{w-b}{w}}q(w)dw = -\frac{2}{3}h_1(z)(z-b)^{3/2},\hspace{0.5cm} z\in\mathcal{U}\cap \gamma^+
\end{equation}
where $h_1(z)$ is an analytic function in $\mathcal{U}$ such that
\begin{equation*}
	h_1(z) = \frac{q(b)}{\sqrt{b}}\bigg[1+\frac{3}{5}\bigg(\frac{q'(b)}{q(b)}-\frac{1}{2b}\bigg)(z-b)+O\big((z-b)\big)^2\bigg],\hspace{0.5cm}z\rightarrow b,
\end{equation*}
and the function $(z-b)^{3/2}$ is defined for $z\in\mathbb{C}\backslash(-\infty,b]$ with
\begin{equation*}
	(z-b)^{3/2}>0\ \ \textnormal{if}\ \ z>b.
\end{equation*}
Similarly, with the same choice of branches,
\begin{equation*}
	G(z) = -2g(z)+V(z)+l = \frac{2}{3}a(z)(z-b)^{3/2},\hspace{0.5cm} z\in\mathcal{U}\cap \gamma^-
\end{equation*}
and
\begin{equation*}
	g_+(z)+g_-(z)-V(z)-l=-\int\limits_b^z\sqrt{\frac{w-b}{w}}q(w)dw=-\frac{2}{3}a(z)(z-b)^{3/2},\ \ \ z\in\mathcal{U}\cap(b,\infty).
\end{equation*}
The expansions motivate the construction of the parametrix in terms of the Airy function $\textnormal{Ai}(\zeta)$. This construction has appeared frequently in the nonlinear-steepest descent literature and we will follow here the notation of \cite{BB}: define for $\zeta\in\mathbb{C}$ 
\begin{equation}\label{pl:2}
	A_0(\zeta) = \begin{pmatrix}
	\frac{d}{d\zeta}\textnormal{Ai}(\zeta) & e^{i\frac{\pi}{3}}\frac{d}{d\zeta}\textnormal{Ai}\Big(e^{-i\frac{2\pi}{3}}\zeta\Big)\\
	\textnormal{Ai}(\zeta) & e^{i\frac{\pi}{3}}\textnormal{Ai}\Big(e^{-i\frac{2\pi}{3}}\zeta\Big)\\
	\end{pmatrix}.
\end{equation}
With this, introduce the ``bare parametrix"
\begin{equation}\label{pl:3}
		A^{RH}(\zeta) = \left\{
                                 \begin{array}{ll}
                                   A_0(\zeta), & \hbox{arg $\zeta\in(0,\frac{2\pi}{3})$,} \smallskip\\
                                   A_0(\zeta)\begin{pmatrix}
                          1 & 0 \\
                          -1 & 1 \\
                        \end{pmatrix}, & \hbox{arg $\zeta\in(\frac{2\pi}{3},\pi)$,} \bigskip \\
                                   A_0(\zeta)\begin{pmatrix}
                                   1 & -1\\
                                   0 & 1\\
                                   \end{pmatrix},& \hbox{arg $\zeta\in(-\frac{2\pi}{3},0)$,} \bigskip \\
                                   A_0(\zeta)\begin{pmatrix}
                                   0 & -1 \\
                                   1 & 1\\
                                   \end{pmatrix}, & \hbox{arg $\zeta\in(-\pi,-\frac{2\pi}{3})$.}
                                 \end{array}
                               \right.
\end{equation}
which solves the RHP depicted in Figure \ref{fig5}
\begin{figure}[tbh]
  \begin{center}
  \psfragscanon
  \psfrag{1}{\footnotesize{$\begin{pmatrix}
  1 & 1\\
  0 & 1\\
  \end{pmatrix}$}}
  \psfrag{2}{\footnotesize{$\begin{pmatrix}
  1 & 0\\
  -1 & 1\\
  \end{pmatrix}$}}
  \psfrag{3}{\footnotesize{$\begin{pmatrix}
  1 & 0\\
  -1 & 1\\
  \end{pmatrix}$}}
  \psfrag{4}{\footnotesize{$\begin{pmatrix}
  0 & -1\\
  1 & 0\\
  \end{pmatrix}$}}
  \includegraphics[width=5cm,height=4cm]{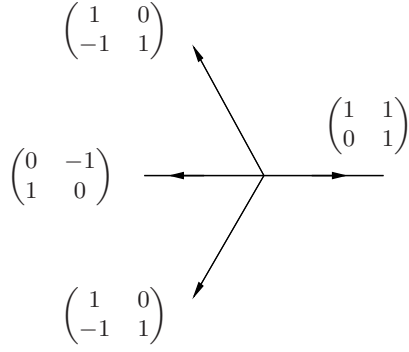}
  \end{center}
  \caption{The model RHP near $z=b$ which can be solved explicitly using Airy functions}
  \label{fig5}
\end{figure}
\begin{itemize}
	\item $A^{RH}(\zeta)$ is analytic for $\zeta\in\mathbb{C}\backslash\{\textnormal{arg}\ \zeta=-\frac{2\pi}{3},0,\frac{2\pi}{3},\pi\}$
	\item We have jumps as sketched in Figure \ref{fig5}
	\begin{eqnarray*}
		A^{RH}_+(\zeta)&=&A^{RH}_-(\zeta)\begin{pmatrix}
		1 & 0\\
		-1 & 1\\
		\end{pmatrix},\hspace{0.5cm}\textnormal{arg}\ \zeta=\mp \frac{2\pi}{3}\\
		A^{RH}_+(\zeta)&=&A^{RH}_-(\zeta)\begin{pmatrix}
		1 & 1\\
		0 & 1\\
		\end{pmatrix},\hspace{0.5cm}\textnormal{arg}\ \zeta=0\\
		A^{RH}_+(\zeta)&=&A^{RH}_-(\zeta)\begin{pmatrix}
		0 & -1\\
		1 & 0\\
		\end{pmatrix},\hspace{0.5cm}\textnormal{arg}\ \zeta=\pi
	\end{eqnarray*}
	\item From the asymptotics of the Airy function (cf. \cite{BE})
	\begin{eqnarray}\label{pl:4}
		A^{RH}(\zeta)&=&\frac{\zeta^{\sigma_3/4}}{2\sqrt{\pi}}\begin{pmatrix}
		-1 & i\\
		1 & i\\
		\end{pmatrix}\Bigg[I+\frac{1}{48\zeta^{3/2}}\begin{pmatrix}
	1 & 6i\\
	6i & -1\\
	\end{pmatrix}+\frac{35}{4608\zeta^{6/2}}\begin{pmatrix}
	-1 & 12i\\
	-12i & -1\\
	\end{pmatrix}\nonumber\\
	&&+O\big(\zeta^{-9/2}\big)\Bigg]e^{-\frac{2}{3}\zeta^{3/2}\sigma_3}.
	\end{eqnarray}
\end{itemize}
In order to construct the local parametrix to the solution of the $S$-RHP near $z=b$, we first define
\begin{equation}\label{pl:5}
	\zeta(z) = \bigg(\frac{3N}{4}\bigg)^{2/3}\Big(-2g(z)+V(z)+l\Big)^{2/3},\hspace{0.5cm}|z-b|<r.
\end{equation}
This change of variables is locally conformal, since
\begin{equation*}
	\zeta(z) = \bigg(\frac{Nq(b)}{2\sqrt{b}}\bigg)^{2/3}(z-b)\bigg[1+\frac{2}{5}\bigg(\frac{q'(b)}{q(b)}-\frac{1}{2b}\bigg)(z-b)+O\big((z-b)^2\big)\bigg],\hspace{0.5cm}|z-b|<r.
\end{equation*}
Secondly, it allows us to define the right parametrix $U(z)$ near $z=b$ by 
\begin{equation}\label{pl:6}
	U(z) = B_r(z)(-i\sqrt{\pi})A^{RH}\big(\zeta(z)\big)e^{\frac{2}{3}\zeta^{3/2}(z)\sigma_3}z^{-\sigma_3/2},\hspace{0.5cm}|z-b|<r
\end{equation}
which involves the multiplier
\begin{eqnarray}
	B_r(z) &=& M(z)z^{\sigma_3/2}\begin{pmatrix}
	-i & i\\
	1 & 1\\
	\end{pmatrix}\zeta^{-\sigma_3/4}(z)
	=\bigg(\frac{\sqrt{b}}{2}\bigg)^{\sigma_3}\begin{pmatrix}
	-i & i\\
	1 & 1\\
	\end{pmatrix}\bigg(\zeta(z)\frac{z}{z-b}\bigg)^{-\sigma_3/4}\nonumber\\
	&&\times\zeta^{\sigma_3/4}(z)\frac{i}{2}\begin{pmatrix}
	1 & -i\\
	-1 & -i\\
	\end{pmatrix}\mathcal{D}^{-\sigma_3}(z)z^{\sigma_3/2}\begin{pmatrix}
	-i & i\\
	1 & 1\\
	\end{pmatrix}\zeta^{-\sigma_3/4}(z).\label{pl:7}
\end{eqnarray}
Notice that $B_r(z)$ is analytic in a neighborhood of $z=b$, since for $z\in(b-r,b)$
\begin{eqnarray*}
	\big(B_r(z)\big)_+ &=& M_+(z)z^{\sigma_3/2}\begin{pmatrix}
	-i & i\\
	1 & 1\\
	\end{pmatrix}\zeta^{-\sigma_3/4}_+(z)\\
	&=& M_-(z) \begin{pmatrix}
	0 & z\\
	-\frac{1}{z} & 0\\
	\end{pmatrix}z^{\sigma_3/2}\begin{pmatrix}
	-i & i\\
	1 & 1\\
	\end{pmatrix}\zeta^{-\sigma_3/4}_-(z)e^{-i\frac{\pi}{2}\sigma_3}=\big(B_r(z)\big)_-,
\end{eqnarray*}
and by a direct computation
\begin{equation*}
	B_r(b) = \bigg(\frac{\sqrt{b}}{2}\bigg)^{\sigma_3}\begin{pmatrix}
	-i & i\\
	1 & 1\\
	\end{pmatrix}\bigg(\frac{Nq(b)}{2}b\bigg)^{-\sigma_3/6}.
\end{equation*}
Thus, after employing a local contour deformation, we derive that the parametrix $U(z)$ has jumps along the curves depicted in Figure \ref{fig6}. Moreover, these jumps are described by the same matrices as in the $S$-RHP, indeed with orientation as in Figure \ref{fig6}
\begin{figure}[tbh]
  \begin{center}
  \psfragscanon
  \psfrag{2}{\footnotesize{$\begin{pmatrix}
  1 & 0\\
  -\frac{1}{z}e^{-NG(z)} & 1\\
  \end{pmatrix}$}}
  \psfrag{4}{\footnotesize{$\begin{pmatrix}
  1 & 0\\
  -\frac{1}{z}e^{NG(z)} & 1\\
  \end{pmatrix}$}}
  \psfrag{1}{\footnotesize{$\begin{pmatrix}
  1 & ze^{N(g_++g_--V-l)}\\
  0 & 1\\
  \end{pmatrix}$}}
  \psfrag{3}{\footnotesize{$\begin{pmatrix}
  0 & -z\\
  \frac{1}{z} & 0\\
  \end{pmatrix}$}}
  \includegraphics[width=6cm,height=3cm]{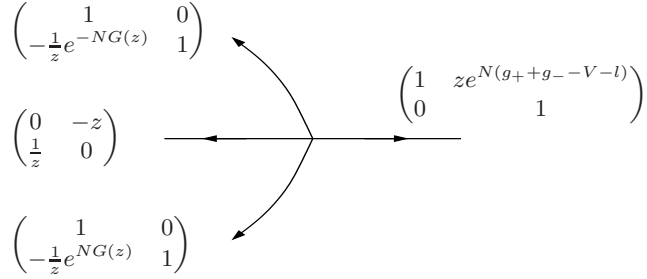}
  \end{center}
  \caption{Transformation of parametrix jumps to original jumps}
  \label{fig6}
\end{figure}
\begin{eqnarray*}
	z^{\sigma_3/2}e^{-\frac{2}{3}\zeta^{3/2}(z)\sigma_3}\begin{pmatrix}
	1 & 0\\
	-1 & 1\\
	\end{pmatrix}e^{\frac{2}{3}\zeta^{3/2}(z)\sigma_3}z^{-\sigma_3/2} &=& \begin{pmatrix}
	1 & 0\\
	-\frac{1}{z}e^{-NG(z)} & 1\\
	\end{pmatrix},\ \ z\in\mathcal{U}\cap \gamma^+\\
	z^{\sigma_3/2}e^{-\frac{2}{3}\zeta^{3/2}(z)\sigma_3}\begin{pmatrix}
	0 & -1\\
	1 & 0\\
	\end{pmatrix}e^{\frac{2}{3}\zeta^{3/2}(z)\sigma_3}z^{-\sigma_3/2}&=& \begin{pmatrix}
	0 & -z\\
	\frac{1}{z} & 0\\
	\end{pmatrix},\ \ z\in\mathcal{U}\cap (0,b)\\
	z^{\sigma_3/2}e^{-\frac{2}{3}\zeta^{3/2}(z)\sigma_3}\begin{pmatrix}
	1 & 1\\
	0 & 1\\
	\end{pmatrix}e^{\frac{2}{3}\zeta^{3/2}(z)\sigma_3}z^{-\sigma_3/2} &=& \begin{pmatrix}
	1 & ze^{N(g_++g_--V-l)}\\
	0 & 1\\
	\end{pmatrix},\ \ z\in\mathcal{U}\cap (b,\infty)\\
	z^{\sigma_3/2}e^{-\frac{2}{3}\zeta^{3/2}(z)\sigma_3}\begin{pmatrix}
	1 & 0\\
	-1 & 1\\
	\end{pmatrix}e^{\frac{2}{3}\zeta^{3/2}(z)\sigma_3}z^{-\sigma_3/2} &=& \begin{pmatrix}
	1 & 0\\
	-\frac{1}{z}e^{NG(z)} & 1\\
	\end{pmatrix},\ \ z\in\mathcal{U}\cap \gamma^-.
\end{eqnarray*}
But this means that the ratio of $S(z)$ with $U(z)$ is locally analytic (here we use the boundedness of the Airy function at the origin), i.e.
\begin{equation}\label{pl:8}
	S(z) = N_r(z)U(z),\ \ |z-b|<r<\frac{b}{2}.
\end{equation}
The use of the multiplier $B_r(z)$ in the definition \eqref{pl:6} follows from the need of a ``matching" between the local model functions $U(z)$ and $M(z)$: observe that
\begin{equation*}
	B_r(z)\Big(-\frac{i}{2}\Big)\zeta^{\sigma_3/4}(z)\begin{pmatrix}
	-1 & i\\
	1 & i\\
	\end{pmatrix} = M(z)z^{\sigma_3/2},
\end{equation*}
so that with the asymptotics \eqref{pl:4},
\begin{eqnarray}
	U(z) &=& M(z)z^{\sigma_3/2}\Bigg[I+\frac{1}{48\zeta^{3/2}}\begin{pmatrix}
	1 & 6i\\
	6i & -1\\
	\end{pmatrix}+\frac{35}{4608\zeta^{6/2}}\begin{pmatrix}
	-1 & 12i\\
	-12i & -1\\
	\end{pmatrix}+O\big(\zeta^{-9/2}\big)\Bigg]z^{-\sigma_3/2}\nonumber\\
	&=&\Bigg[I+\frac{U_1(z)}{96\zeta^{3/2}}+\frac{35U_2(z)}{4608\zeta^{6/2}}+O\big(\zeta^{-9/2}\big)\Bigg]M(z)\label{pl:9}
\end{eqnarray}
as $N\rightarrow\infty$ for any $\al>1$ and $0<r_1\leq |z-b|\leq r_2<\frac{b}{2}$ (so $|\zeta|\rightarrow\infty$). Here $U_k = (U_k^{ij})$ are given by
\begin{eqnarray*}
	U_1^{11}(z) &=& \delta^2(z)\bigg(1-\frac{3}{z}\mathcal{D}^2(z)-3z\mathcal{D}^{-2}(z)\bigg)+\delta^{-2}(z)\bigg(1+\frac{3}{z}\mathcal{D}^2(z)+3z\mathcal{D}^{-2}(z)\bigg) = -U_1^{22}\\
	U_1^{12}(z) &=&-\frac{ib}{4}\bigg(\delta^2(z)-\delta^{-2}(z)-\frac{3}{z}\mathcal{D}^2(z)\big(\delta(z)-\delta^{-1}(z)\big)^2\\
	&&-3z\mathcal{D}^{-2}(z)\big(\delta(z)+\delta^{-1}(z)\big)^2\bigg)
\end{eqnarray*}
\begin{eqnarray*}
	U_1^{21}(z) &=& -\frac{4i}{b}\bigg(\delta^2(z)-\delta^{-2}(z)-\frac{3}{z}\mathcal{D}^2(z)\big(\delta(z)+\delta^{-1}(z)\big)^2\\
	&&-3z\mathcal{D}^{-2}(z)\big(\delta(z)-\delta^{-1}(z)\big)^2\bigg)
\end{eqnarray*}
and
\begin{eqnarray*}
	U_2^{11}(z) &=& -1+\delta^2(z)\bigg(\frac{3}{z}\mathcal{D}^2(z)-3z\mathcal{D}^{-2}(z)\bigg)-\delta^{-2}(z)\bigg(\frac{3}{z}\mathcal{D}^2(z)-3z\mathcal{D}^{-2}(z)\bigg)\\
	U_2^{22}(z) &=& -1-\delta^2(z)\bigg(\frac{3}{z}\mathcal{D}^2(z)-3z\mathcal{D}^{-2}(z)\bigg)+\delta^{-2}(z)\bigg(\frac{3}{z}\mathcal{D}^2(z)-3z\mathcal{D}^{-2}(z)\bigg)\\
	U_2^{12}(z) &=& -\frac{ib}{4}\bigg(\frac{3}{z}\mathcal{D}^2(z)\big(\delta(z)-\delta^{-1}(z)\big)^2-3z\mathcal{D}^{-2}(z)\big(\delta(z)+\delta^{-1}(z)\big)^2\bigg)\\
	U_2^{21}(z) &=& -\frac{4i}{b}\bigg(\frac{3}{z}\mathcal{D}^2(z)\big(\delta(z)+\delta^{-1}(z)\big)^2-3z\mathcal{D}^{-2}(z)\big(\delta(z)-\delta^{-1}(z)\big)^2\bigg).
\end{eqnarray*}
But as the function $\zeta(z)$ is of order $N^{2/3}$ on the latter annulus and $\delta(z),\mathcal{D}(z)$ are bounded, we obtain an asymptotical matching between the model functions from equation \eqref{pl:8},
\begin{equation*}
	U(z) = \big(I+o(1)\big)M(z),\hspace{0.5cm}N\rightarrow\infty,\ \ 0\leq\tau\leq 1-\varepsilon<1,\ \ 0<r_1\leq |z-b|\leq r_2<\frac{b}{2}
\end{equation*}
The latter relation will be important later on and we also emphasize that the last estimation is uniform with respect to the parameter $0\leq\tau\leq 1-\varepsilon<1$. 


\subsection{Construction of a parametrix at $z=0$}

Fix a small neighborhood $\mathcal{V}$ of the origin and observe that
\begin{eqnarray*}
	G(z) = 2g(z)-V(z)-l &=& \int\limits_z^b\sqrt{\frac{w-b}{w}}q(w)dw = 2\pi i-\int\limits_0^z\sqrt{\frac{w-b}{w}}q(w)dw\\
	&=&2\pi i-2h_2(z)\sqrt{z},\hspace{0.5cm}z\in\mathcal{V}\cap \gamma^+,\ \ 0<\textnormal{arg}\ z\leq 2\pi
\end{eqnarray*}
where $h_2(z)$ is analytic in $\mathcal{V}$ such that
\begin{equation*}
	h_2(z)=e^{i\frac{\pi}{2}}q(0)\sqrt{b}\bigg[1+\frac{1}{3}\bigg(\frac{q'(0)}{q(0)}-\frac{1}{2b}\bigg)z+O\big(z^2\big)\bigg],\hspace{0.5cm}z\rightarrow 0.
\end{equation*}
Similarly 
\begin{equation*}
	G(z) = -2g(z)+V(z)+l=2\pi i+2h_2(z)\sqrt{z},\hspace{0.5cm}z\in\mathcal{V}\cap\gamma^-,\ \ 0<\textnormal{arg}\ z\leq 2\pi
\end{equation*}
and both stated local behaviors suggest to use the Bessel functions $I_1(\zeta)$ and $K_1(\zeta)$ in the construction of an edge point parametrix. Again, we proceed in several steps. First we recall (cf. \cite{BE}) that the modified Bessel functions are unique independent solutions to Bessel's equation
\begin{equation*}
	z^2w''+zw'-(z^2+1)w = 0
\end{equation*}
satisfying the following asymptotic conditions as $\zeta\rightarrow\infty$ and $-\frac{\pi}{2}<\textnormal{arg}\ \zeta<\frac{3\pi}{2}$
\begin{equation*}
	I_1(\zeta)\sim\frac{e^{\zeta}}{\sqrt{2\pi\zeta}}\bigg(1-\frac{3}{8\zeta}-\frac{15}{64\zeta^2}+O\big(\zeta^{-3}\big)\bigg)+
	\frac{e^{-\zeta}e^{i\frac{3\pi}{2}}}{\sqrt{2\pi\zeta}}\bigg(1+\frac{3}{8\zeta}-\frac{15}{64\zeta^2}+O\big(\zeta^{-3}\big)\bigg)
\end{equation*} 
as well as for $-\frac{3\pi}{2}<\textnormal{arg}\ \zeta<\frac{\pi}{2}$
\begin{equation*}
		I_1(\zeta)\sim\frac{e^{\zeta}}{\sqrt{2\pi\zeta}}\bigg(1-\frac{3}{8\zeta}-\frac{15}{64\zeta^2}+O\big(\zeta^{-3}\big)\bigg)+
	\frac{e^{-\zeta}e^{-i\frac{3\pi}{2}}}{\sqrt{2\pi\zeta}}\bigg(1+\frac{3}{8\zeta}-\frac{15}{64\zeta^2}+O\big(\zeta^{-3}\big)\bigg).
\end{equation*}
On the other hand
\begin{equation*}
	K_1(\zeta) = \sqrt{\frac{\pi}{2\zeta}}e^{-\zeta}\bigg(1+\frac{3}{8\zeta}-\frac{15}{64\zeta^2}+O\big(\zeta^{-3}\big)\bigg),\hspace{0.5cm}\zeta\rightarrow\infty
\end{equation*}
which holds in a full neighborhood of infinity. Secondly $I_1(\zeta),K_1(\zeta)$ satisfy monodromy relations, valid on the entire universal covering of the punctured plane
\begin{equation}\label{pr:1}
	I_1\big(e^{-i\pi }\zeta\big) = e^{-i\pi }I_1(\zeta),\hspace{1cm} K_1(e^{-i\pi }\zeta) = e^{i\pi }K_1(\zeta)+i\pi I_1(\zeta)
\end{equation}
and finally the following expansions at the origin are valid
\begin{equation*}
	I_1(\zeta) = \frac{\zeta}{2}\bigg(1+\frac{\zeta^2}{8}+O\big(\zeta^4\big)\bigg),\hspace{0.5cm} K_1(\zeta)=\frac{1}{\zeta}\Big(1+O\big(\zeta^2\ln\zeta\big)\Big),\ \ \zeta\rightarrow 0.
\end{equation*}
Remembering the latter properties we introduce on the punctured plane $\zeta\in\mathbb{C}\backslash\{0\}$,
\begin{equation}\label{pr:2}
	P_{BE}(\zeta) = e^{-i\frac{\pi}{4}}\begin{pmatrix}
	I_1(2e^{-i\frac{\pi}{2}}\sqrt{\zeta}) & -\frac{i}{\pi}K_1(2e^{-i\frac{\pi}{2}}\sqrt{\zeta})\\
	-2\pi i\sqrt{\zeta}\big(I_1\big)'(2e^{-i\frac{\pi}{2}}\sqrt{\zeta}) & -2\sqrt{\zeta}\big(K_1\big)'(2e^{-i\frac{\pi}{2}}\sqrt{\zeta})\\
	\end{pmatrix},\hspace{0.5cm}0<\textnormal{arg}\ \zeta\leq2\pi.
\end{equation}
From the behavior of $I_1(\zeta)$ and $K_1(\zeta)$ at infinity, we deduce
\begin{eqnarray}
	P_{BE}(\zeta) &=& \zeta^{-\sigma_3/4}(2\pi)^{-\sigma_3/2}\frac{1}{\sqrt{2}}\begin{pmatrix}
	1 & -i\\
	-i & 1\\
	\end{pmatrix}\Bigg[I+\frac{1}{16\sqrt{\zeta}}\begin{pmatrix}
	-5i & -2\\
	-2 & 5i\\
	\end{pmatrix} +\frac{3}{64\zeta}\begin{pmatrix}-1& -6i\\
	6i & -1\\
	\end{pmatrix}\nonumber\\
	&&+O\big(\zeta^{-3/2}\big)\Bigg]e^{-2i\sqrt{\zeta}\sigma_3}\begin{pmatrix}
	1 & 0\\
	1 & 1\\
	\end{pmatrix},\hspace{0.5cm}\zeta\rightarrow\infty,\ \frac{\pi}{3}<\textnormal{arg}\ \zeta<\frac{7\pi}{3}\label{pr:3}
\end{eqnarray}
and
\begin{eqnarray}
	P_{BE}(\zeta) &=& \zeta^{-\sigma_3/4}(2\pi)^{-\sigma_3/2}\frac{1}{\sqrt{2}}\begin{pmatrix}
	1 & -i\\
	-i & 1\\
	\end{pmatrix}\Bigg[I+\frac{1}{16\sqrt{\zeta}}\begin{pmatrix}
	-5i & -2\\
	-2 & 5i\\
	\end{pmatrix} +\frac{3}{64\zeta}\begin{pmatrix}-1& -6i\\
	6i & -1\\
	\end{pmatrix}\nonumber\\
	&&+O\big(\zeta^{-3/2}\big)\Bigg]e^{-2i\sqrt{\zeta}\sigma_3}\begin{pmatrix}
	1 & 0\\
	-1 & 1\\
	\end{pmatrix},\hspace{0.5cm}\zeta\rightarrow\infty,\ -\frac{\pi}{3}<\textnormal{arg}\ \zeta<\frac{5\pi}{3}.\label{pr:4}
\end{eqnarray}
The ``bare parametrix" is given by
\begin{equation}\label{pr:5}
	P_{BE}^{RH}(\zeta) = \left\{
                                 \begin{array}{ll}
                                   P_{BE}(\zeta)\begin{pmatrix}
                                   1 & 0\\
                                   1 & 1\\
                                   \end{pmatrix}, & \hbox{arg $\zeta\in(0,\frac{\pi}{3})$,} \smallskip\\
                                   P_{BE}(\zeta), & \hbox{arg $\zeta\in(\frac{\pi}{3},\frac{5\pi}{3})$,} \smallskip \\
                                   P_{BE}(\zeta)\begin{pmatrix}
                                   1 & 0\\
                                   -1 & 1\\
                                   \end{pmatrix},& \hbox{arg $\zeta\in(\frac{5\pi}{3},2\pi)$.} 
                                 \end{array}
                               \right.
\end{equation}
and its analytical properties summarized below
\begin{figure}[tbh]
  \begin{center}
  \psfragscanon
  \psfrag{1}{\footnotesize{$\begin{pmatrix}
  0 & 1\\
  -1 & 0\\
  \end{pmatrix}$}}
  \psfrag{2}{\footnotesize{$\begin{pmatrix}
  1 & 0\\
  -1 & 1\\
  \end{pmatrix}$}}
  \psfrag{3}{\footnotesize{$\begin{pmatrix}
  1 & 0\\
  -1 & 1\\
  \end{pmatrix}$}}
  \includegraphics[width=3cm,height=4cm]{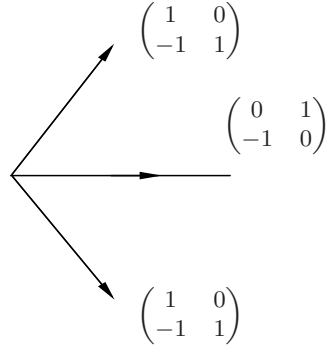}
  \end{center}
  \caption{The model RHP near $z=0$ which can be solved explicitly using Bessel functions}
  \label{fig7}
\end{figure}

\begin{itemize}
	\item $P_{BE}^{RH}(\zeta)$ is analytic for $\zeta\in\mathbb{C}\backslash\big\{\textnormal{arg}\ \zeta=0,\frac{\pi}{3},\frac{5\pi}{3}\big\}$
	\item The following jumps hold, see Figure \ref{fig7},
	\begin{eqnarray*}
		\big(P_{BE}^{RH}(\zeta)\big)_+ &=& \big(P_{BE}^{RH}(\zeta)\big)_-\begin{pmatrix}
		1 & 0\\
		-1 & 1\\
		\end{pmatrix},\hspace{1cm}\textnormal{arg}\ \zeta=\frac{\pi}{3}\\
		\big(P_{BE}^{RH}(\zeta)\big)_+ &=& \big(P_{BE}^{RH}(\zeta)\big)_-\begin{pmatrix}
		1 & 0\\
		-1 & 1\\
		\end{pmatrix},\hspace{1cm}\textnormal{arg}\ \zeta=\frac{5\pi}{3}.
	\end{eqnarray*}
	For the jump on the line $\textnormal{arg}\ \zeta = 0$ we notice that the monodromy relations imply
	\begin{eqnarray*}
		I_1(2e^{-i\frac{\pi}{2}}\sqrt{\zeta}_+) &=& I_1(2e^{-i\frac{\pi}{2}}e^{-i\pi}\sqrt{\zeta}_-) = e^{-i\pi}I_1(2e^{-i\frac{\pi}{2}}\sqrt{\zeta}_-)\\
		\big(I_1\big)'(2e^{-i\frac{\pi}{2}}\sqrt{\zeta}_+) &=& \big(I_1\big)'(2e^{-i\frac{\pi}{2}}\sqrt{\zeta}_-)
	\end{eqnarray*}
	and
	\begin{eqnarray*}
		K_1(2e^{-i\frac{\pi}{2}}\sqrt{\zeta}_+) &=& K_1(2e^{-i\frac{\pi}{2}}e^{-i\pi}\sqrt{\zeta}_-) = e^{i\pi}K_1(2e^{-i\frac{\pi}{2}}\sqrt{\zeta}_-)+i\pi I_1(2e^{-i\frac{\pi}{2}}\sqrt{\zeta}_-)\\
		\big(K_1\big)'(2e^{-i\frac{\pi}{2}}\sqrt{\zeta}_+) &=& \big(K_1\big)'(2e^{-i\frac{\pi}{2}}\sqrt{\zeta}_-)-i\pi\big(I_1\big)'(2e^{-i\frac{\pi}{2}}\sqrt{\zeta}_-).
	\end{eqnarray*}
	Therefore
	\begin{equation*}
		\big(P_{BE}(\zeta)\big)_+=\big(P_{BE}(\zeta)\big)_-\begin{pmatrix}
		-1 & 1\\
		0 & -1\\
		\end{pmatrix}
	\end{equation*}
	and hence
	\begin{equation*}
		\big(P_{BE}^{RH}(\zeta)\big)_+=\big(P_{BE}^{RH}(\zeta)\big)_-\begin{pmatrix}
		0 & 1\\
		-1 & 0\\
		\end{pmatrix},\hspace{1cm}\textnormal{arg}\ \zeta=0.
	\end{equation*}
	\item At the origin 
	\begin{equation}\label{pr:6}
		P_{BE}^{RH}(\zeta) = \frac{e^{-i\frac{\pi}{4}}}{2\sqrt{\zeta}}\Bigg[\begin{pmatrix}
		0 & \frac{1}{\pi}\\
		0 & 1\\
		\end{pmatrix}+O\big(\zeta\ln\zeta\big)\Bigg],\hspace{0.5cm}\zeta\rightarrow 0,\ \ \frac{\pi}{3}<\textnormal{arg}\ \zeta<\frac{5\pi}{3}
	\end{equation}
	and for the other sector we have to multiply the latter expansion with the correct multipliers from \eqref{pr:5}
	\item In order to determine the behavior of $P_{BE}^{RH}(\zeta)$ at infinity, we recall \eqref{pr:3} and \eqref{pr:4} as well as
	\begin{equation*}
		e^{-2i\sqrt{\zeta}\sigma_3}\begin{pmatrix}
		1 & 0\\
		\pm 1 & 1\\
		\end{pmatrix}e^{2i\sqrt{\zeta}\sigma_3} = \begin{pmatrix}
		1 & 0\\
		\pm e^{4i\sqrt{\zeta}} & 1\\
		\end{pmatrix},\hspace{0.5cm}\frac{\pi}{3}<\textnormal{arg}\ \zeta<\frac{5\pi}{3}.
	\end{equation*}
	However for those $\zeta$, we have $\textnormal{Re}\big(4i\sqrt{\zeta}\big)<0$, hence the given product approaches the identity exponentially fast as $\zeta\rightarrow\infty$. Together we have
	\begin{eqnarray}\label{pr:7}
		P_{BE}^{RH}(\zeta)&=& \zeta^{-\sigma_3/4}(2\pi)^{-\sigma_3/2}\frac{1}{\sqrt{2}}\begin{pmatrix}
	1 & -i\\
	-i & 1\\
	\end{pmatrix}\Bigg[I+\frac{1}{16\sqrt{\zeta}}\begin{pmatrix}
	-5i & -2\\
	-2 & 5i\\
	\end{pmatrix} \nonumber\\
	&&+\frac{3}{64\zeta}\begin{pmatrix}-1& -6i\\
	6i & -1\\
	\end{pmatrix}
	+O\big(\zeta^{-3/2}\big)\Bigg]e^{-2i\sqrt{\zeta}\sigma_3},\hspace{0.5cm}\zeta\rightarrow\infty 
\end{eqnarray}
valid in a whole neighborhood of infinity. 
\end{itemize}
With the help of the model function $P_{BE}^{RH}(\zeta)$, the local parametrix near $z=0$ is now defined as follows: first define
\begin{equation}\label{pr:8}
	\zeta(z) = e^{-i\pi}\bigg(\frac{N}{4}\bigg)^2\Big(-2g(z)+V(z)+l-2\pi i\Big)^2,\hspace{0.5cm} |z|<r,\ \ 0<\textnormal{arg}\ \zeta\leq 2\pi.
\end{equation}
which is also a locally conformal change of variables, as
\begin{equation*}
	\zeta(z) = \bigg(\frac{Nq(0)\sqrt{b}}{2}\bigg)^2z\bigg[1+\frac{2}{3}\bigg(\frac{q'(0)}{q(0)}-\frac{1}{2b}\bigg)z+O\big(z^2\big)\bigg],\hspace{0.5cm}|z|<r.
\end{equation*}
Using the change $\zeta=\zeta(z)$, the left parametrix $V(z)$ near $z=0$ is given by the formula
\begin{equation}\label{pr:9}
	W(z) = B_l(z)P_{BE}^{RH}\big(\zeta(z)\big)e^{2i\zeta^{1/2}(z)\sigma_3}(-z)^{-\sigma_3/2},\hspace{0.5cm} |z|<r
\end{equation}
with the matrix multiplier
\begin{equation*}
	B_l(z) =M(z)(-z)^{\sigma_3/2}\frac{1}{\sqrt{2}}\begin{pmatrix}
	1 & i\\
	i & 1\\
	\end{pmatrix}\zeta^{\sigma_3/4}(z)(2\pi)^{\sigma_3/2}.
\end{equation*}
Again $B_l(z)$ is analytic in a neighborhood of $z=0$, for $z\in(0,r)$
\begin{eqnarray*}
	\big(B_l(z)\big)_+ &=&M_+(z)(-z)^{\sigma_3/2}_+\frac{1}{\sqrt{2}}\begin{pmatrix}
	1& i\\
	i & 1\\
	\end{pmatrix}\zeta_+^{\sigma_3/4}(z)(2\pi)^{\sigma_3/2}\\
	&=& M_-(z)\begin{pmatrix}
	0 & z\\
	-\frac{1}{z} & 0\\
	\end{pmatrix}(-z)_-^{\sigma_3/2}e^{i\pi\sigma_3}\frac{1}{\sqrt{2}}\begin{pmatrix}
	1 & i\\
	i & 1\\
	\end{pmatrix}\zeta_-^{\sigma_3/4}(z)e^{-i\frac{\pi}{2}\sigma_3}(2\pi)^{\sigma_3/2}\\
	&=& \big(B_l(z)\big)_-
\end{eqnarray*}
and with
\begin{equation*}
	\delta(z)\big(\zeta(z)\big)^{-1/4} = \bigg(\frac{z}{\zeta(z)(z-b)}\bigg)^{1/4} = \bigg(\frac{4e^{i\pi}}{N^2b^2q(0)}\bigg)^{1/4}\Big(1+O\big(z\big)\Big),\hspace{0.5cm}z\rightarrow 0
\end{equation*}
we obtain from a direct computation
\begin{equation*}
	B_l(0) = \frac{1}{\sqrt{2}}\bigg(\frac{\sqrt{b}}{2}\bigg)^{\sigma_3}\begin{pmatrix}
	1 & i\\
	i & 1\\
	\end{pmatrix}\bigg(\frac{N^2b^2q(0)}{4}e^{-i\pi}\bigg)^{\sigma_3/4}.
\end{equation*}
Moreover the latter identity combined with \eqref{pr:6}, allows us to show that 
\begin{equation*}
	W(z) = O(1),\hspace{0.5cm} z\rightarrow 0,\ \ \ \frac{\pi}{3}<\textnormal{arg}\ z<\frac{5\pi}{3}
\end{equation*}
as well as
\begin{equation*}
	W(z) = O\big(z^{-1}\big),\hspace{0.5cm} z\rightarrow 0,\ \ \ 0<\textnormal{arg}\ z<\frac{\pi}{3},\ \ \frac{5\pi}{3}<\textnormal{arg}\ z<2\pi,
\end{equation*}
which precisely matches the endpoint behavior of $S(z)$ in \eqref{opeq2} and \eqref{opeq3}. On the other hand, the parametrix $W(z)$ has jumps along the curves depicted in Figure \ref{fig8}, and we can locally match the latter curves with the jump contour in the S-RHP.\\[0.10cm]
\smallskip
\begin{figure}[tbh]
  \begin{center}
  \psfragscanon
  \psfrag{1}{\footnotesize{$\begin{pmatrix}
  0 & z\\
  -\frac{1}{z} & 1\\
  \end{pmatrix}$}}
  \psfrag{2}{\footnotesize{$\begin{pmatrix}
  1 & 0\\
  \frac{1}{z}e^{-NG(z)} & 1\\
  \end{pmatrix}$}}
  \psfrag{4}{\footnotesize{$\begin{pmatrix}
  1 & e^{N(g_++g_--V-l)}\\
  0 & 1\\
  \end{pmatrix}$}}
  \psfrag{4}{\footnotesize{$\begin{pmatrix}
  1 & 0\\
  \frac{1}{z}e^{NG(z)} & 1\\
  \end{pmatrix}$}}
  \includegraphics[width=7cm,height=3.5cm]{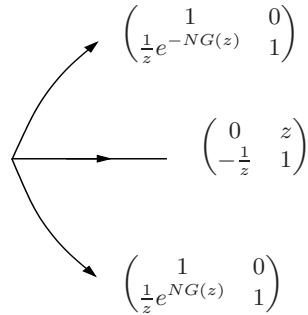}
  \end{center}
  \caption{Transformation of parametrix jumps to original jumps}
  \label{fig8}
\end{figure}

Moreover, the jumps are identical to the ones in the $S$-RHP since
\begin{eqnarray*}
	(-z)^{\sigma_3/2}e^{-2i\zeta^{1/2}(z)\sigma_3}\begin{pmatrix}
	1 & 0\\
	-1 & 1\\
	\end{pmatrix}e^{2i\zeta^{1/2}(z)\sigma_3}(-z)^{-\sigma_3/2} &=& \begin{pmatrix}
	1 & 0\\
	\frac{1}{z}e^{-NG(z)} & 1\\
	\end{pmatrix},\hspace{0.5cm}z\in\mathcal{V}\cap\gamma^+\\
	(-z)^{\sigma_3/2}e^{-2i\zeta^{1/2}(z)\sigma_3}\begin{pmatrix}
	0 & 1\\
	-1 & 0\\
	\end{pmatrix}e^{2i\zeta^{1/2}(z)\sigma_3}(-z)^{-\sigma_3/2} &=& \begin{pmatrix}
	0 & z\\
	-\frac{1}{z} & 0\\
	\end{pmatrix},\hspace{0.5cm}z\in\mathcal{V}\cap(0,b)\\
	(-z)^{\sigma_3/2}e^{-2i\zeta^{1/2}(z)\sigma_3}\begin{pmatrix}
	1 & 0\\
	-1 & 1\\
	\end{pmatrix}e^{2i\zeta^{1/2}(z)\sigma_3}(-z)^{-\sigma_3/2} &=& \begin{pmatrix}
	1 & 0\\
	\frac{1}{z}e^{NG(z)} & 1\\
	\end{pmatrix},\hspace{0.5cm}z\in\mathcal{V}\cap\gamma^-.
\end{eqnarray*}
Hence the ratio of $S(z)$ with $W(z)$ is locally analytic, i.e.
\begin{equation}\label{pr:10}
	S(z) = N_l(z)W(z),\hspace{0.5cm} 0<|z|<r<\frac{b}{2}.
\end{equation}
The role of the left multiplier $B_l(z)$ in \eqref{pr:9} is the same as in the construction of the parametrix $U(z)$, it provides us with an asymptotic matching relation between the model functions: with
\begin{equation*}
	B_l(z)\zeta^{-\sigma_3/4}(z)(2\pi)^{-\sigma_3/2}\frac{1}{\sqrt{2}}\begin{pmatrix}
	1 & -i\\
	-i & 1\\
	\end{pmatrix} = M(z)(-z)^{\sigma_3/2}
\end{equation*}
we deduce from \eqref{pr:7}
\begin{eqnarray}\label{pr:11}
	W(z)&=&M(z)(-z)^{\sigma_3/2}\Bigg[I+\frac{1}{16\sqrt{\zeta}}\begin{pmatrix}
	-5i & -2\\
	-2 & 5i\\
	\end{pmatrix}+\frac{3}{64\zeta}\begin{pmatrix}-1& -6i\\
	6i & -1\\
	\end{pmatrix}
	+O\big(\zeta^{-3/2}\big)\Bigg](-z)^{-\sigma_3/2}\nonumber\\
	&=&\Bigg[I+\frac{W_1(z)}{32\sqrt{\zeta}}+\frac{3W_2(z)}{64\zeta}+O\big(\zeta^{-3/2}\big)\Bigg]M(z)
\end{eqnarray}
as $N\rightarrow\infty$ (hence $|\zeta|\rightarrow\infty$), for any $\al>1$ and $0<r_1\leq|z|\leq r_2<\frac{b}{2}$. The coefficients $W_k=(W_k^{ij})$ are given by
\begin{eqnarray*}
	W_1^{11}(z) &=&-i\delta^2(z)\bigg(5-\frac{1}{z}\mathcal{D}^2(z)-z\mathcal{D}^{-2}(z)\bigg)-i\delta^{-2}(z)\bigg(5+\frac{1}{z}\mathcal{D}^2(z)+z\mathcal{D}^{-2}(z)\bigg)\\
	 &=& -W_1^{22}(z),\\
	W_1^{12}(z) &=&\frac{b}{4}\bigg(-5\big(\delta^2(z)-\delta^{-2}(z)\big)+\frac{1}{z}\mathcal{D}^2(z)\big(\delta(z)-\delta^{-1}(z)\big)^2\\
	&&+z\mathcal{D}^{-2}(z)\big(\delta(z)+\delta^{-1}(z)\big)^2\bigg),\\
	W_1^{21}(z) &=& \frac{4}{b}\bigg(-5\big(\delta^2(z)-\delta^{-2}(z)\big)+\frac{1}{z}\mathcal{D}^2(z)\big(\delta(z)+\delta^{-1}(z)\big)^2\\
	&&+z\mathcal{D}^{-2}(z)\big(\delta(z)-\delta^{-1}(z)\big)^2\bigg),
\end{eqnarray*}
and
\begin{eqnarray*}
	W_2^{11}(z)&=&-1+\frac{3}{2}\delta^2(z)\bigg(\frac{1}{z}\mathcal{D}^2(z)-z\mathcal{D}^{-2}(z)\bigg)-\frac{3}{2}\delta^{-2}(z)\bigg(\frac{1}{z}\mathcal{D}^2(z)-z\mathcal{D}^{-2}(z)\bigg),\\
	W_2^{22}(z)&=&-1-\frac{3}{2}\delta^2(z)\bigg(\frac{1}{z}\mathcal{D}^2(z)-z\mathcal{D}^{-2}(z)\bigg)+\frac{3}{2}\delta^{-2}(z)\bigg(\frac{1}{z}\mathcal{D}^2(z)-z\mathcal{D}^{-2}(z)\bigg),\\
	W_2^{12}(z)&=&-\frac{ib}{4}\bigg(\frac{3}{2z}\mathcal{D}^2(z)\big(\delta(z)-\delta^{-1}(z)\big)^2-\frac{3}{2}z\mathcal{D}^{-2}(z)\big(\delta(z)+\delta^{-1}(z)\big)^2\bigg),\\
	W_2^{21}(z)&=&-\frac{4i}{b}\bigg(\frac{3}{2z}\mathcal{D}^2(z)\big(\delta(z)+\delta^{-1}(z)\big)^2-\frac{3}{2}z\mathcal{D}^{-2}(z)\big(\delta(z)-\delta^{-1}(z)\big)^2\bigg).
\end{eqnarray*}
Since $\zeta(z)$ is of order $N^2$ on the latter annulus and $\delta(z),\mathcal{D}(z)$ are bounded, equation \eqref{pr:10} implies the following matching relation between $W(z)$ and $M(z)$,
\begin{equation}\label{pr:12}
	W(z) = \big(I+o(1)\big)M(z),\hspace{0.5cm} N\rightarrow\infty,\ 0\leq\tau\leq 1-\varepsilon<1,\ \ 0<r_1\leq |z|\leq r_2<\frac{b}{2}.
\end{equation}
We now use the model functions $M(z),U(z)$ and $W(z)$ and employ another transformation.


\subsection{Third transformation of the RHP - ratio problem}

In this step we put
\begin{equation}\label{R:1}
	R(z) = S(z)\left\{
                                   \begin{array}{ll}
                                     \big(W(z)\big)^{-1}, & \hbox{$|z|<r_1$,} \\
                                     \big(U(z)\big)^{-1}, & \hbox{$|z-b|<r_2$,} \\
                                     \big(M(z)\big)^{-1}, & \hbox{$|z|>r_1,|z-b|>r_2$} 
                                   \end{array}
                                 \right.
\end{equation}
with $0<r_2<\frac{b}{2}$ and $0\leq r_1<\min\big\{\frac{\pi}{2t},\frac{b}{2}\big\}$. The reason for choosing the latter radius in this explicit $t$-dependent form arises from the analyticity of the potential $V(z)$, which is holomorphic in the strip $\Delta_t$. Moreover the set of its branch points $\Omega_t$ is given by
\begin{equation*}
	\Omega_t = \left\{i\frac{n\pi}{t}:\ n\in\mathbb{Z}\backslash\{0\}\right\}.
\end{equation*}
Hence we need to choose a neighborhood of the origin in \eqref{R:1} which does not include any of the branch points. With $C_{t,b}$ denoting the clockwise oriented circles shown in Figure \ref{fig9}, the ratio-function $R(z)$ solves the following RHP
\begin{figure}[tbh]
  \begin{center}
  \psfragscanon
  \psfrag{1}{\footnotesize{$C_t$}}
  \psfrag{2}{\footnotesize{$C_b$}}
  \psfrag{4}{\footnotesize{$\hat{\gamma}^+$}}
  \psfrag{5}{\footnotesize{$\hat{\gamma}^-$}}
  \includegraphics[width=8cm,height=3cm]{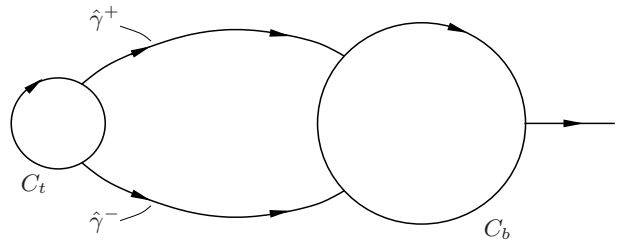}
  \end{center}
  \caption{The jump graph for the ratio function $R(z)$}
  \label{fig9}
\end{figure}
\begin{itemize}
	\item $R(z)$ is analytic for $z\in\mathbb{C}\backslash\big\{C_{t,b}\cup\hat{\Gamma}\cup(b+r_2,\infty)\big\}$ with $\hat{\Gamma} = \hat{\gamma}_+\cup\hat{\gamma}_-$
	\item The jumps are as follows: on the infinite branch $(b+r_2,\infty)$,
	\begin{equation*}
		R_+(z) = R_-(z)M(z)\begin{pmatrix}
		1 & ze^{N(g_++g_--V-l)}\\
		0 & 1\\
		\end{pmatrix}\big(M(z)\big)^{-1},
	\end{equation*}
		on the upper lens boundary $\hat{\gamma}_+$, resp. lower lens boundary $\hat{\gamma}_-$,
		\begin{eqnarray*}
			R_+(z) &=&R_-(z)M(z)\begin{pmatrix}
			1 & 0\\
			\frac{1}{z}e^{-NG(z)} & 1 \\
			\end{pmatrix}\big(M(z)\big)^{-1},\hspace{0.5cm}z\in\hat{\gamma}_+\\
			R_+(z) &=&R_-(z)M(z)\begin{pmatrix}
			1 & 0\\
			\frac{1}{z}e^{NG(z)} & 1 \\
			\end{pmatrix}\big(M(z)\big)^{-1},\hspace{0.5cm}z\in\hat{\gamma}_-
		\end{eqnarray*}
		and on the clockwise oriented circles $C_{t,b}$,
		\begin{equation*}
			R_+(z) = R_-(z)\left\{
                                   \begin{array}{ll}
                                     W(z)\big(M(z)\big)^{-1}, & \hbox{$|z|=r_1$,} \\
                                     U(z)\big(M(z)\big)^{-1}, & \hbox{$|z-b|=r_2$.} 
                                   \end{array}
                                 \right.
   \end{equation*}
	\item As $z\rightarrow\infty$, the function is normalized as $R(z)\rightarrow I$.
\end{itemize}
We note that $R(z)$ has no jumps inside of $C_t$ and $C_b$ and across the line segment in between. Also, $R(z)$ is bounded at $z=0$ and $z=b$, which follows from \eqref{pl:8} and \eqref{pr:10}. To move ahead, we recall the previously stated behavior of the jump matrices as $N\rightarrow\infty$ and note that these estimations are valid for any $\al>1$ such that $0\leq\tau\leq1-\varepsilon<1$. In fact, on the half ray $(b+r_2,\infty)$, the jumps approach the identity matrix. Through \eqref{pl:9}, the same is also true on the circle $C_b$, more precisely with $G_R$ denoting the jump matrix in the latter ratio-RHP
\begin{equation}\label{R:2}
	\|G_R-I\|_{L^2\cap L^{\infty}(C_b)}\leq \frac{\hat{c}}{N},\hspace{0.5cm} N\rightarrow\infty,\ \ 0\leq\tau\leq1-\varepsilon<1
\end{equation}
with a constant $\hat{c}>0$ whose value is not important. For the lens boundaries we use the local identities
\begin{equation*}
	G(z) = 2\pi i\mp2h_2(z)\sqrt{z},\ \ z\in\mathcal{V}\cap\gamma^{\pm}.
\end{equation*}
They imply
\begin{equation}\label{R:3}
	\sup_{z\in\hat{\gamma}^+\cup\hat{\gamma}^-}\big|G_R(z)-I\big| = \sup_{z\in(\hat{\gamma}^+\cup\hat{\gamma}^-)\cap C_t}\big|G_R(z)-I\big| = O\big(te^{-cNt^{-1/2}}\big),\hspace{0.25cm}N\rightarrow\infty,\ \ 0\leq\tau\leq1-\varepsilon<1
\end{equation}
and since
\begin{equation*}
	\frac{N}{\sqrt{t}} = \sqrt{N\al} \rightarrow\infty\hspace{1cm}\textnormal{as}\ N\rightarrow\infty\hspace{0.5cm}\textnormal{for any}\ \al>1:\ \ 0\leq\tau\leq1-\varepsilon<1,
\end{equation*}
we see that the contributions arising from the lenses decay exponentially fast. In order to estimate $G_R$ on the circle $C_t$, we use \eqref{pr:11}
\begin{equation}\label{R:4}
	\sup_{z\in C_t}\bigg|G_R(z)-I-\frac{W_1(z)}{32\sqrt{\zeta(z)}}\bigg| = O\bigg(\frac{t}{N^2}\bigg) = O\big(\tau N^{-1}\big),\hspace{0.5cm} N\rightarrow\infty,\ \ 0\leq\tau\leq1-\varepsilon<1.
\end{equation}
This estimation holds since \eqref{pr:11} extends to a full asymptotic series of the form
\begin{equation}\label{R:5}
	G_R(z)-I = \sum_{k=1}^{\infty}\tilde{W}_k(z)N^{-k},
\end{equation}
valid as long as $N^2|z|\rightarrow\infty$. However for $z\in C_t$
\begin{equation*}
	\frac{1}{N^2|z|} = O\big(\tau N^{-1}\big)\hspace{0.5cm}\textnormal{as}\ N\rightarrow\infty\hspace{0.5cm}\textnormal{for any}\ \al>1:\ 0\leq\tau\leq1-\varepsilon<1,
\end{equation*}
so \eqref{R:5} holds in particular for $z\in C_t$ and since $\tilde{W}_k(z)$ has a pole of order at most $\left\lfloor\frac{k+1}{2}\right\rfloor$ at the origin, we obtain \eqref{R:4}. From the local expansions
\begin{equation*}
	\mathcal{D}^2(z) = -z\big(1+O\left(\sqrt{z}\,\right)\big),\hspace{0.5cm} \delta^2(z) = -i\sqrt{b}\sqrt{z}\big(1+O(z)\big),\ \ \ z\rightarrow 0,\ \ z\in\mathbb{C}\backslash[0,b]
\end{equation*}
we now evaluate the residue of $\tilde{W}_1(z)$ at the origin
\begin{equation*}
	B = \textnormal{res}_{z=0} \tilde{W}_1(z) = \textnormal{res}_{z=0}\bigg(\frac{NW_1(z)}{32\sqrt{\zeta(z)}}\bigg) = \frac{3}{16q(0)}\bigg(\frac{b}{4}\bigg)^{\sigma_3/2}\begin{pmatrix}
	1 & e^{i\frac{\pi}{2}}\\
	e^{i\frac{\pi}{2}} & -1\\
	\end{pmatrix}\bigg(\frac{b}{4}\bigg)^{-\sigma_3/2}.
\end{equation*}
Hence we can rewrite \eqref{R:4} as
\begin{equation}\label{R:6}
	\sup_{z\in C_t}\bigg|G_R(z)-I-\frac{B}{Nz}\bigg|=O\big(N^{-1}\big),\hspace{0.5cm} N\rightarrow\infty,\ \ 0\leq\tau\leq1-\varepsilon<1,
\end{equation}
however 
\begin{equation*}
	\sup_{z\in C_t}\bigg|\frac{B}{Nz}\bigg| = O(\tau),\hspace{0.5cm}N\rightarrow\infty,\ \ 0\leq\tau\leq1-\varepsilon<1,
\end{equation*}
so $G_R(z)-I$ is not uniformly close to zero on $C_t$ as $N\rightarrow\infty$ for all $0\leq\tau\leq1-\varepsilon<1$. To overcome this difficulty we employ our final transformation.


\subsection{Fourth and final transformation of the RHP}

Since $\det B= \textnormal{trace}\ B = 0$, we see that the matrix function $I+\frac{B}{Nz}$ is unimodular for any $z\in\mathbb{C}\backslash\{0\}$, in fact
\begin{equation*}
	\bigg(I+\frac{B}{Nz}\bigg)^{-1}=I-\frac{B}{Nz},\hspace{0.5cm}z\in\mathbb{C}\backslash\{0\}.
\end{equation*} 
We introduce 
\begin{equation}\label{U:1}
	Q(z) = \left\{
                                   \begin{array}{ll}
                                     R(z), & \hbox{$|z|\leq r_1$,} \\
                                     R(z)\big(I+\frac{B}{Nz}\big)^{-1}, & \hbox{$|z|>r_1$.} 
                                   \end{array}
                                 \right.
\end{equation}
and are lead to the following RHP
\begin{itemize}
	\item $Q(z)$ is analytic for $z\in\mathbb{C}\backslash\big\{C_{t,b}\cup\hat{\Gamma}\cup(b+r_2,\infty)\big\}$
	\item With $G_Q$ denoting the jump matrix in the $Q$-RHP we have
	\begin{eqnarray*}
		G_Q(z)&=&G_R(z)\bigg(I+\frac{B}{Nz}\bigg)^{-1},\hspace{0.5cm} z\in C_t\\
		G_Q(z)&=&\bigg(I+\frac{B}{Nz}\bigg)G_R(z)\bigg(I+\frac{B}{Nz}\bigg)^{-1},\hspace{0.5cm} z\in C_b\cup\hat{\Gamma}\cup(b+r_2,\infty)
	\end{eqnarray*}
	\item As $z\rightarrow\infty$, we have $Q(z)\rightarrow I$
\end{itemize}
In the stated problem all jump matrices approach the identity matrix as $N\rightarrow\infty$ for any $\al>1$ such that $0\leq\tau\leq1-\varepsilon<1$. More precisely with $\Sigma_Q$ denoting the underlying contour
\begin{equation}\label{U:2}
	\|G_Q-I\|_{L^2\cap L^{\infty}(\Sigma_Q)}\leq \frac{c}{N},\hspace{0.5cm} N\rightarrow\infty,\ \ 0\leq\tau\leq1-\varepsilon<1
\end{equation}
with a constant $c>0$ whose value is not important. This estimation allows us to solve the $Q$-RHP iteratively.


\subsection{Solution of the RHP for $Q(z)$ via iteration}

The final RHP for the function $Q(z)$ reads as
\begin{itemize}
	\item $Q(z)$ is analytic for $z\in\mathbb{C}\backslash\Sigma_Q$.
	\item The boundary values on the contour shown in Figure \ref{fig9} are related via the identity
	\begin{equation*}
		Q_+(z) = Q_-(z)G_Q(z),\hspace{0.5cm} z\in\Sigma_Q.
	\end{equation*}
	\item The normalization $Q(z)=I+O(z^{-1})$ is valid as $z\rightarrow\infty$
\end{itemize}
and it is equivalent to the singular integral equation
\begin{equation}\label{It:1}
	Q_-(z) = I+\frac{1}{2\pi i}\int\limits_{\Sigma_Q}Q_-(w)\big(G_Q(w)-I\big)\frac{dw}{w-z_-}.
\end{equation}
Through \eqref{U:2}, we obtain \cite{DZ,BK} that equation \eqref{It:1} can be solved iteratively in $L^2(\Sigma_Q)$ for sufficiently large $N$ and $\al>1:0\leq\tau\leq1-\varepsilon<1$. 
Also, the unique solution satisfies
\begin{equation}\label{It:2}
	\|Q_--I\|_{L^2(\Sigma_Q)}\leq \frac{c}{N},\hspace{0.5cm} N\rightarrow\infty\hspace{0.5cm}\textnormal{for any}\ \al>1:\ 0\leq\tau\leq1-\varepsilon<1.
\end{equation}
We are now ready to determine the large $N$ asymptotics of the normalizing constants $h_{N,t}$. To this end notice that for $z\in\mathbb{C}\backslash\Sigma_Q$
\begin{equation}\label{It:3}
	Q(z) = I+\frac{i}{2\pi z}\int\limits_{\Sigma_Q}Q_-(w)\big(G_Q(w)-I\big)dw +O\big(z^{-2}\big),\hspace{0.5cm}z\rightarrow\infty
\end{equation}
and also as $N\rightarrow\infty$ for any $\al>1:0\leq\tau\leq1-\varepsilon<1$ (from \eqref{U:2} and \eqref{It:2} as well as the previous discussion about exponentially small contributions)
\begin{equation}\label{It:4}
	\int\limits_{\Sigma_Q}Q_-(w)\big(G_Q(w)-I\big)dw = \int\limits_{C_t}\big(G_Q(w)-I\big)dw +\int\limits_{C_b}\big(G_Q(w)-I\big)dw +O\left(N^{-2}\right).
\end{equation}


\section{Asymptotics of $h_{N,t}$ - proof of theorem \ref{theo1}}\label{proof_theo1}

We go back to \eqref{FIKeq1}
\begin{equation*}
	h_{N,t} = -2\pi i\left(Y_1^{(N)}\right)_{12}
\end{equation*}
and recall that
\begin{equation}\label{pr1eq1}
	Y_1^{(N)} = \lim_{z\rightarrow\infty}\left(z\left(Y^{(N)}(z)z^{-N\sigma_3}-I\right)\right).
\end{equation}
Now recall the sequence of transformations
\begin{equation*}
	Y(z)\equiv Y^{(N)}(z)\mapsto T(z)\mapsto S(z)\mapsto R(z)\mapsto Q(z)
\end{equation*}
and combine it with the expansion
\begin{equation*}
	e^{N(g(z)-\frac{l}{2})\sigma_3}z^{-N\sigma_3} = e^{-\frac{Nl}{2}\sigma_3}\left(I-\frac{N\sigma_3}{z}\int\limits_0^bw\psi(w)dw+O\left(z^{-2}\right)\right),\hspace{0.5cm}z\rightarrow\infty.
\end{equation*}
This gives us for \eqref{pr1eq1}
\begin{equation*}
	Y_1^{(N)} = \lim_{z\rightarrow\infty}\Big(z\big(e^{\frac{Nl}{2}\sigma_3}Q(z)\bigg(I+\frac{B}{Nz}\bigg)M(z)e^{N(g(z)-\frac{l}{2})\sigma_3}z^{-N\sigma_3}-I\big)\Big)
\end{equation*}
and hence (compare \eqref{It:3})
\begin{eqnarray*}
	e^{-\frac{Nl}{2}\sigma_3}Y_1^{(N)}e^{\frac{Nl}{2}\sigma_3} &=& -N\sigma_3\int\limits_0^bw\psi(w)dw +\frac{b}{4}\begin{pmatrix}
	-1 & \frac{ib}{4}\\
	\frac{4}{ib} & 1\\
	\end{pmatrix} +\frac{B}{N}\\
	&&+\frac{i}{2\pi}\int\limits_{\Sigma_Q}Q_-(w)\big(G_Q(w)-I\big)\,dw.
\end{eqnarray*}
In view of \eqref{It:4}, we will now compute the contribution from the circle $C_b$. First from \eqref{pl:9} as $N\rightarrow\infty$ for any $\al>1: 0\leq\tau\leq1-\varepsilon<1$
\begin{equation*}
	\int\limits_{C_b}Q_-(w)\big(G_Q(w)-I\big)dw = \frac{1}{96}\int\limits_{C_b}\begin{pmatrix}
	U_{11}(w) & U_{12}(w)\\
	U_{21}(w) & U_{22}(w)\\
	\end{pmatrix}\frac{dw}{\zeta^{3/2}(w)} +O\left(N^{-2}\right).
\end{equation*}
Now use the local expansions
\begin{eqnarray*}
	U_1^{11}(z) &=& -\frac{5\sqrt{b}}{\sqrt{z-b}}\left(1+\frac{3}{2b}(z-b)+O\left((z-b)^2\right)\right),\\
	U_1^{12}(z) &=&\frac{5ib\sqrt{b}}{4\sqrt{z-b}}\left(1+\frac{91}{10b}(z-b)+O\left((z-b)^2\right)\right),\\
	U_1^{21}(z) &=&\frac{20i}{\sqrt{b}\sqrt{z-b}}\left(1-\frac{1}{2b}(z-b)+O\left((z-b)^2\right)\right),
\end{eqnarray*}
valid as $z\rightarrow b$, and with \eqref{pl:5} compute the relevant line integral via residue theorem. We obtain, as $N\rightarrow\infty$,
\begin{eqnarray*}
	\int\limits_{C_b}Q_-(w)\big(G_Q(w)-I\big)dw&=&-\frac{2\pi i}{96}\frac{b}{Nq(b)}
\begin{pmatrix}
	6(\frac{q'(b)}{q(b)}-\frac{3}{b}) & \frac{ib}{2}\big[-3\frac{q'(b)}{q(b)}+\frac{47}{b}\big]\\
	\frac{8i}{b}\big[-3\frac{q'(b)}{q(b)}-\frac{1}{b}\big] & -6(\frac{q'(b)}{q(b)}-\frac{3}{b})\\
	\end{pmatrix}\\
	&&+O\left(N^{-2}\right)
\end{eqnarray*}
which is uniform with respect to the parameter $0\leq\tau\leq1-\varepsilon<1$. For the remaining line integral along the circle boundary $C_t$ recall \eqref{pr:11} and \eqref{R:6} and deduce, as $N\rightarrow\infty$
\begin{equation}\label{pr1eq2}
	\int\limits_{C_t}Q_-(w)\big(G_Q(w)-I\big)dw = \int\limits_{C_t}\Big(G_R(w)-I-\frac{B}{Nw}\Big)\Big(I-\frac{B}{Nw}\Big)dw +O\left(N^{-2}\right)
\end{equation}
which is again uniform with respect to the parameter $0\leq\tau\leq1-\varepsilon<1$. Now from \eqref{pr:11}, as $z\rightarrow 0$,
\begin{eqnarray*}
	W_1^{11}(z)&=&\frac{3\sqrt{b}}{\sqrt{z}}\left(1-\frac{3z}{2b}+O\left(z^2\right)\right),\hspace{0.5cm}
	W_1^{12}(z)=\frac{b}{4}\frac{3i\sqrt{b}}{\sqrt{z}}\left(1+\frac{35z}{6b}+O\left(z^2\right)\right)\\
	W_1^{21}(z)&=&\frac{4}{b}\frac{3i\sqrt{b}}{\sqrt{z}}\left(1+\frac{z}{2b}+O\left(z^2\right)\right),
\end{eqnarray*}
hence
\begin{eqnarray*}
	\frac{W_1(z)}{32\sqrt{\zeta(z)}}-\frac{B}{Nz} &=& -\frac{1}{3N}\bigg(\frac{q'(0)}{q(0)}-\frac{1}{2b}\bigg)B+\frac{1}{2Nb}\frac{3}{16q(0)}\bigg(\frac{b}{4}\bigg)^{\sigma_3/2}\begin{pmatrix}
	-3 & \frac{35i}{3}\\
	i & 3\\
	\end{pmatrix}\bigg(\frac{b}{4}\bigg)^{-\sigma_3/2}\\
	&&+O\left(zN^{-1}\right),\hspace{0.5cm}z\rightarrow 0.
\end{eqnarray*}
Back to \eqref{pr1eq2}, as $N\rightarrow\infty$ therefore
\begin{equation*}
	\int\limits_{C_t}Q_-(w)\big(G_Q(w)-I\big)dw= O\left(N^{-2}\right),
\end{equation*}
which is uniform with respect to the parameter $0\leq\tau\leq1-\varepsilon<1$. At this point we summarize our computations
\begin{eqnarray*}
	e^{-\frac{Nl}{2}\sigma_3}Y_1^{(N)}e^{\frac{Nl}{2}\sigma_3} &=& -N\sigma_3\int\limits_0^bw\psi(w)dw +\frac{b}{4}\begin{pmatrix}
	-1 & \frac{ib}{4}\\
	\frac{4}{ib} & 1\\
	\end{pmatrix}+\frac{B}{N}\\
	&&+\frac{1}{96}\frac{b}{Nq(b)}\begin{pmatrix}
	6(\frac{q'(b)}{q(b)}-\frac{3}{b}) & \frac{ib}{2}\big[-3\frac{q'(b)}{q(b)}+\frac{47}{b}\big]\\
	\frac{8i}{b}\big[-3\frac{q'(b)}{q(b)}-\frac{1}{b}\big] & -6(\frac{q'(b)}{q(b)}-\frac{3}{b})\\
	\end{pmatrix}+O\left(N^{-2}\right),
\end{eqnarray*}
which implies, as $N\rightarrow\infty$, 
\begin{equation*}
	\left(Y_1^{(N)}\right)_{12} = ie^{Nl}\bigg(\frac{b}{4}\bigg)^2\left[1+\frac{v}{N}+O\left(N^{-2}\right)\right],\hspace{0.5cm}0\leq\tau\leq1-\varepsilon<1,
\end{equation*}
with (compare \eqref{the14})
\begin{equation*}
	v = \frac{3}{4bq(0)}-\frac{q'(b)}{4q^2(b)}+\frac{47}{12bq(b)}.
\end{equation*}
All we need to do now is recall \eqref{FIKeq1}, the connection formula $h_N=2t^{2N+2}h_{N,t}$ and combine it with Stirling's approximation
\begin{equation*}
	N! = \left(\frac{N}{e}\right)^N\sqrt{2\pi N}\left(1+\frac{1}{12N}+O\left(N^{-2}\right)\right),\hspace{0.5cm}N\rightarrow\infty.
\end{equation*}
This gives, as $N\rightarrow\infty$,
\begin{equation}\label{pr1eq3}
	\frac{h_N}{(N!)^2} = \frac{N}{8}\tau^{2N+2}b^2\exp\left[N(l+2)+\frac{v}{N}-\frac{1}{6N}+\varepsilon_N(\tau)\right]
\end{equation}
which is uniform with respect to the parameter $0\leq\tau\leq1-\varepsilon<1$, thus proving Theorem \ref{theo1}.\smallskip

\begin{rem} At this point it is useful to compare the latter expansion to the estimation \eqref{BLres1} derived in \cite{BL2}. We obtain from the connection $h_N^o=(\al-1)^{2N+1}h_N$ and \eqref{pr1eq3}, that, as $N\rightarrow\infty$,
\begin{equation}\label{pr1eq4}
	\ln\left[\frac{h_N^o}{(N!)^2}\right] = (2N+1)\ln(1-\tau)+N(l+2)+\ln\left(\frac{t}{8}\right)+2\ln b+\frac{v}{N}-\frac{1}{6N}+\varepsilon_N(\tau)
\end{equation}
uniformly with respect to $0\leq\tau\leq1-\varepsilon<1$. Also, as a consequence of the Riemann-Hilbert analysis presented in the last subsections, the estimation
\begin{equation*}
	\left|\varepsilon_N\right|\leq \frac{c}{(N+1)^2},\hspace{0.5cm}c>0
\end{equation*}
on the error term $\varepsilon_N(\tau)$, can in fact be extended to a full asymptotic series in reciprocal integer powers of $N$ which is also uniform with respect to the parameter $0\leq\tau\leq1-\varepsilon<1$. Now choose $\al$ from any compact subset of the set \eqref{excset1} and let $N\rightarrow\infty$, i.e. $t\rightarrow\infty$. In this limit, Proposition \ref{prop2} implies with \eqref{prop1eq2},
\begin{equation*}
	b = \frac{4}{1-\tau}\left(1-\frac{1}{2N}+\frac{\zeta(3/2)}{8\sqrt{\pi(r-1)}N^{3/2}}+O\left(N^{-5/2}\right)\right),\hspace{0.5cm} N\rightarrow\infty,\ \ r = \frac{\al+1}{\al-1}
\end{equation*}
which extends to a full asymptotic series in reciprocal half-integer powers of $N$, the error terms being uniform on any compact subset of the set \eqref{excset1}. Also via \eqref{lageq5} and \eqref{prop4eq2}, as $N\rightarrow\infty$,
\begin{equation*}
	l = 4(1-\ln 2)-\frac{3b}{2}(1-\tau)+2\ln b-\frac{2}{N}(1-\ln 2)-\frac{\ln(2bN\tau)}{N}+O\left(N^{-5/2}\right)
\end{equation*}
and which can also be extended to a full asymptotic series in reciprocal half-integer powers of $N$. Combining the last two expansions,
\begin{equation*}
	l = -2-4\ln 2+2\ln\left(\frac{4}{1-\tau}\right)-\frac{\ln(2N)}{N}+\frac{\ln(\al-1)}{N}-\frac{\zeta(3/2)}{2\sqrt{\pi(r-1)}N^{3/2}}+\frac{1}{4N^2}+O\left(N^{-5/2}\right).
\end{equation*}
and since from \eqref{denseq3}, as $N\rightarrow\infty$
\begin{equation*}
	v = \frac{7}{6}+O\left(N^{-1/2}\right),
\end{equation*}
we can go back to \eqref{pr1eq4} and derive
\begin{equation}\label{pr1eq5}
	\ln\left[\frac{h_N^o}{(N!)^2}\right] = -\frac{\zeta(3/2)}{2\sqrt{\pi(r-1)}N^{1/2}}+\frac{1}{4N} +O\left(N^{-3/2}\right),\hspace{0.5cm}N\rightarrow\infty
\end{equation}
which is uniform on any compact subset of the set \eqref{excset1}. The last estimation agrees with \eqref{BLres1} and as we have seen, extends to a full asymptotic series in reciprocal half-integer powers of $N$.
\end{rem}
As a first step in the computation of the $N$ independent leading term $C$ in the large $N$ expansion \eqref{the20} of $Z_N$, we use the Toda equation.


\section{Toda equation and the structure of the constant factor}\label{toda_eq}

We use the Toda equation as written in \eqref{zj5},
\begin{equation*}
	\left(\ln\tau_N\right)'' = \frac{h_N}{h_{N-1}} = \frac{h_N^o}{(\al-1)^2h_{N-1}^o},\hspace{0.5cm}(') = \frac{d}{d\alpha}.
\end{equation*}
From \eqref{pr1eq5} and our discussion thereafter, as $N\rightarrow\infty$,
\begin{equation*}
	\ln\left[\frac{h_N^o}{(N!)^2}\right] = -\frac{\zeta(3/2)}{2\sqrt{\pi(r-1)}N^{1/2}}+\frac{1}{4N} +\frac{c_1(\alpha)}{N^{3/2}}+\frac{c_2(\al)}{N^2}+O\left(N^{-5/2}\right)
\end{equation*}
with some constants $c_i(\al)$ whose precise form is not important for us. Hence
\begin{equation*}
	\ln\left[\frac{h_N^o}{N^2h_{N-1}^o}\right] = \frac{\zeta(3/2)}{4\sqrt{\pi(r-1)}N^{3/2}}-\frac{1}{4N^2}+O\left(N^{-5/2}\right)
\end{equation*}
and after exponentiating the latter expansion, as $N\rightarrow\infty$
\begin{eqnarray*}
	\left(\ln\tau_N\right)'' &=&\frac{N^2}{(\al-1)^2}\left(1+\frac{\zeta(3/2)}{4\sqrt{\pi(r-1)}N^{3/2}}-\frac{1}{4N^2}+O\left(N^{-5/2}\right)\right)\\
	&=& - N^2\left(\ln(\al-1)\right)''-\sqrt{N}\,\frac{\zeta(3/2)}{\sqrt{2\pi}}\left(\sqrt{\al-1}\,\right)''+\frac{1}{4}\left(\ln(\al-1)\right)''+O\left(N^{-1/2}\right),
\end{eqnarray*}
where the error term is uniform on any compact subset of the set \eqref{excset1}. Back to \eqref{zj4}, we have therefore shown that
\begin{equation}\label{todaeq1}
	\left(\ln Z_N\right)'' = N^2\left(\ln\left(\frac{\al+1}{2}\right)\right)''-\sqrt{N}\,\frac{\zeta(3/2)}{\sqrt{2\pi}}\left(\sqrt{\al-1}\,\right)''
	+\frac{1}{4}\left(\ln(\al-1)\right)''+O\left(N^{-1/2}\right).
\end{equation}
On the other hand from \eqref{zj4} combined with \eqref{pr1eq5}, 
\begin{equation*}
	\ln Z_N = \ln C+N^2\ln\left(\frac{\al+1}{2}\right)-\sqrt{N}\,\frac{\zeta(3/2)}{\sqrt{2\pi}}\sqrt{\al-1}+\frac{1}{4}\ln N+O\left(N^{-1/2}\right),
\end{equation*}
where $C>0$ depends in general on $\al$, but not on $N$. Thus, comparing the latter with \eqref{todaeq1}, we conclude
\begin{equation*}
	\ln C = \frac{1}{4}\left(\ln(\al-1)\right)'' +O\left(N^{-1/2}\right).
\end{equation*}
Integrating this expansion, we get
\begin{equation}\label{todaeq2}
	\ln C = \frac{1}{4}\ln(\al-1)+d(N)\al+c(N)+O\left(N^{-1/2}\right)
\end{equation}
with some numbers $d(N)$ and $c(N)$ which are independent of $\al$. Now choose any distinct $\al_1,\al_2$ from \eqref{excset1} and derive
\begin{equation*}
	\ln C(\al_1)-\ln C(\al_2)=\frac{1}{4}\ln(\al_1-1)-\frac{1}{4}\ln(\al_2-1)+d_1(N)(\al_1-\al_2)+O\left(N^{-1/2}\right),
\end{equation*}
which shows that the limit
\begin{equation*}
	\lim_{N\rightarrow\infty}d(N) = d
\end{equation*}
exists and therefore also the limit
\begin{equation*}
	\lim_{N\rightarrow\infty}c(N) = c.
\end{equation*}
Taking the limit $N\rightarrow\infty$ in \eqref{todaeq2}, we obtain
\begin{equation*}
	\ln C = \frac{1}{4}\ln(\al-1)+d\al+c,
\end{equation*}
and summarize (see \eqref{the21})
\begin{prop}\label{prop5} The constant factor $C$ in asymptotic formula \eqref{the20} has the form
\begin{equation}\label{todaeq3}
	C =  (\al-1)^{1/4}e^{d\al +c}.
\end{equation}
\end{prop}
In light of the last proposition we now have to compute the remaining two universal constants $c$ and $d$. 
This will be done by studying two regimes of the double scaling parameter $t=\frac{N}{\al}$. 
First, we are interested in the behavior of the partition function $Z_N$ as $N\rightarrow\infty$ and $t$ remains bounded.


\section{The double scaling limit of the partition function}\label{dsl}

We start with the observation that 
\begin{equation*}
	\lim_{\substack{\al\rightarrow\infty\\ N\leq N_0}}w_t(x) = xe^{-Nx},
\end{equation*}
which in particular implies
\begin{equation*}
	w_t(x)\sim xe^{-Nx}\equiv w_0(x),\hspace{0.5cm}t\rightarrow 0.
\end{equation*}
The limiting orthogonal polynomials are the normalized (and rescaled) Laguerre polynomials (cf. \cite{BE})
\begin{equation*}
	p_{n,0}(x) = \lim_{t\rightarrow 0}p_{n,t}(x) = \frac{(-1)^nn!}{N^n}L_n^{(1)}(Nx)
\end{equation*}
for which
\begin{equation*}
	h_{n,t} \sim \int\limits_0^{\infty}\left(p_{n,0}(x)\right)^2w_0(x)\,dx = \frac{(n!)^2}{N^{2n+2}}\int\limits_0^{\infty}\left(L_n^{(1)}(x)\right)^2xe^{-x}\,dx = \frac{(n!)^2(n+1)}{N^{2n+2}}\equiv h_{n,0},\hspace{0.5cm}t\rightarrow 0.
\end{equation*}
Let us introduce the abbreviation
\begin{equation}\label{dbeq1}
	\sigma_{N,t} = N^{N(N+1)}\prod_{k=0}^{N-1}\frac{h_{k,t}}{(k!)^2},
\end{equation}
which satisfies
\begin{equation}\label{dbeq2}
	\lim_{\substack{\al\rightarrow\infty\\ N\leq N_0}} \sigma_{N,t} = N^{N(N+1)}\prod_{k=0}^{N-1}\frac{h_{k,0}}{(k!)^2} = N!
\end{equation}
and which relates to the partition function $Z_N$ via the identity
\begin{equation}\label{dbeq3}
	Z_N = \left(\frac{\al^2-1}{2\al}\right)^{N^2}\left(\frac{2}{\al}\right)^N\sigma_{N,t}.
\end{equation}
We will now evaluate \eqref{dbeq1} by using \eqref{pr1eq3}, in other words
\begin{eqnarray*}
	\sigma_{N,t} &=& N^{N(N+1)}h_{0,t}\prod_{k=1}^{N-1}\frac{h_{k,t}}{(k!)^2} = N^2h_{0,t}\left(\frac{\al}{2}\right)^{N-1}\al^{N^2-1}\prod_{k=1}^{N-1}\frac{h_k}{(k!)^2}\\
	&=&N^2h_{0,t}(N-1)!\,\exp\left[\sum_{k=1}^{N-1}\left(2\ln\left(\frac{b}{4}\right)+k(l+2)+\frac{v-\frac{7}{6}}{k}+\frac{1}{k}+\varepsilon_k(\tau)\right)\right]\\
	&=&\hat{C}_0\,N^2h_{0,t}\,N!\exp\left[\sum_{k=1}^{N-1}\left(2\ln\left(\frac{b}{4}\right)+k(l+2)+\frac{v-\frac{7}{6}}{k}\right)\right]\left(1+O\left(N^{-1}\right)\right),
\end{eqnarray*}
valid as $N\rightarrow\infty$, where the error term is uniform with respect to the parameter $0\leq\tau\leq1-\varepsilon<1$ and with a universal, i.e. $N$ and $\tau$ independent constant $\hat{C}_0>0$. Now use \eqref{zj0} and derive
\begin{equation*}
	h_{0,t} = \int\limits_0^{\infty}\left(p_{0,t}(x)\right)^2w_t(x)\,dx = \hat{c}_0\int\limits_0^{\infty}w_t(x)\,dx = \frac{\hat{c}_0}{N^2(1-\tau^2)}
\end{equation*}
with another universal constant $\hat{c}_0>0$. Back to the previous expansion for $\sigma_{N,t}$, as $N\rightarrow\infty$
\begin{equation}\label{dbeq4}
	\sigma_{N,t} = \frac{C_0\,N!}{1-\tau^2}\,\exp\left[\sum_{k=1}^{N-1}\left(2\ln\left(\frac{b}{4}\right)+k(l+2)+\frac{v-\frac{7}{6}}{k}\right)\right]\left(1+O\left(N^{-1}\right)\right)
\end{equation}
which is uniform with respect to the parameter $0\leq\tau\leq1-\varepsilon<1$. In order to determine the constant $C_0$, we will now evaluate the sums in \eqref{dbeq4} in the double scaling limit $N,\al\rightarrow\infty$ with $0\leq t\leq t_0$ and then compare the result with \eqref{dbeq2}.\smallskip

For the sums, use Euler's summation formula
\begin{equation}\label{euler}
	\sum_{k=1}^{N-1}g(k) = \int\limits_1^{N-1}g(x)\,dx +\int\limits_1^{N-1}P_1(x)g'(x)\,dx+\frac{1}{2}\left(g(N-1)+g(1)\right),
\end{equation}
which holds for a differentiable function $g:\mathbb{R}\rightarrow\mathbb{R}$ with the Bernoulli polynomial $P_1(x)=x-\left\lfloor x\right\rfloor -\frac{1}{2}$. First via \eqref{prop1eq1}, as $N,\al\rightarrow\infty$,
\begin{equation*}
	\sum_{k=1}^{N-1}2\ln\left(\frac{b}{4}\right) = 2t+\int\limits_0^tI(8x)\frac{dx}{x}+O\left(N^{-1}\right),
\end{equation*}
where the error term is uniform on any finite interval $0\leq t\leq t_0$. Secondly via \eqref{lageq5}
\begin{eqnarray*}
	\sum_{k=1}^{N-1}k(l+2) &=&N\left[-t+\frac{2}{t}\int\limits_0^t\big(J(8x)-(1-\ln 2)\big)\,dx-\frac{2}{t}\int\limits_0^tI(8x)\,dx+\frac{1}{t}\int\limits_0^t\ln S(4x)\,dx\right]\\
	&&+t-\frac{3}{2}t^2-2tI(8t)-\frac{1}{2}I(8t)-\frac{1}{2}I^2(8t)-\frac{3}{2}I(8t)-\frac{1}{4}\int\limits_0^tI^2(8x)\frac{dx}{x}\\
	&&+2\int\limits_0^tJ'(8x)\left(4+\frac{2}{x}I(8x)\right)\,dx+J(8t)-(1-\ln 2)+\frac{1}{2}\ln S(4t)+O\left(N^{-1}\right)
\end{eqnarray*}
and from \eqref{denseq3},
\begin{equation*}
	\sum_{k=1}^{N-1}\frac{v-\frac{7}{6}}{k} = O\left(N^{-1}\right)
\end{equation*}
as $N,\al\rightarrow\infty$, where the error terms are uniform on any finite interval $0\leq t\leq t_0$. We go back to \eqref{dbeq4}, as $N,\al\rightarrow\infty$,
\begin{equation}\label{dbeq5}
	\sigma_{N,t} = C_0\,N!\,e^{N\Phi(t)+\Psi(t)}\left(1+O\left(N^{-1}\right)\right),\hspace{0.5cm}0\leq t\leq t_0
\end{equation}
with
\begin{equation}\label{dbeq6}
	\Phi(t) = -t+\frac{2}{t}\int\limits_0^t\big(J(8x)-(1-\ln 2)\big)\,dx-\frac{2}{t}\int\limits_0^tI(8x)\,dx+\frac{1}{t}\int\limits_0^t\ln S(4x)\,dx
\end{equation}
and
\begin{eqnarray}\label{dbeq7}
	\Psi(t) &=& 3t-\frac{3}{2}t^2-2tI(8t)-\frac{1}{2}I^2(8t)-\frac{3}{2}I(8t)-\frac{1}{2}\ln S(4t)-\big(J(8t)-(1-\ln 2)\big)\\
	&&-\frac{1}{4}\int\limits_0^tI^2(8x)\frac{dx}{x}+\int\limits_0^tI(8x)\frac{dx}{x}+2\int\limits_0^t\left(J'(8x)
	+\frac{S'(4x)}{2S(4x)}\right)\left(1+\frac{I(8x)}{2x}\right)4x\,dx.\nonumber
\end{eqnarray}
The small $t$-behavior of $\Phi(t)$ and $\Psi(t)$ can be determined from \eqref{prop1eq1} and \eqref{prop4eq1}, we have, as $t\rightarrow 0$,
\begin{equation*}
	\Phi(t) = \frac{t^2}{3}+O\left(t^3\right),\hspace{0.5cm} \Psi(t) = 5t-\frac{9}{2}t^2+O\left(t^3\right).
\end{equation*}
Back to \eqref{dbeq2}, we have on one hand
\begin{equation}\label{dbeq8}
	\sigma_{N,t} \sim N!\hspace{0.5cm}\textnormal{as}\ \ t\rightarrow 0.
\end{equation}
On the other hand, if we let $N,\al\rightarrow\infty$ such that $Nt^2\rightarrow 0$ (i.e. in particular $t\rightarrow 0$), then \eqref{dbeq5} and the behavior of $\Phi(t)$ and $\Psi(t)$ at the origin imply, that
\begin{equation}\label{dbeq9}
	\sigma_{N,t} \sim C_0\,N!\hspace{0.5cm}\textnormal{as}\ \ N,\al\rightarrow\infty:\ Nt^2\rightarrow 0.
\end{equation}
Comparing \eqref{dbeq8} with \eqref{dbeq9}, this implies
\begin{equation*}
	C_0 = 1,
\end{equation*}
and we have therefore shown
\begin{theo}\label{prop6} In the double scaling limit $N,\al\rightarrow\infty$
\begin{equation}\label{prop6eq1}
	Z_N=N!\,\left(\frac{\al^2-1}{2\al}\right)^{N^2}\left(\frac{2}{\al}\right)^Ne^{N\Phi(t)+\Psi(t)}\left(1+O\left(N^{-1}\right)\right),
\end{equation}
where $\Phi(t)$ and $\Psi(t)$ are given explicitly in \eqref{dbeq6}, \eqref{dbeq7} and the error term is uniform on any finite interval $0\leq t\leq t_0$.
\end{theo}

The explicit evaluation of the numerical constant $C_0$ is crucial for our further strategy. In order to compute the constants $c$ and $d$, we will go back to \eqref{dbeq4}, evaluate now the sums in the limit $t\rightarrow\infty$ and then compare the result with \eqref{the20} and \eqref{todaeq3}.


\section{Proof of theorem \ref{theorem2}}\label{proof_theo2}
The computations in the last section lead to the following expansion for $\sigma_{N,t}$, as $N\rightarrow\infty$
\begin{equation*}
	\sigma_{N,t} = \frac{N!}{1-\tau^2}\,\exp\left[\sum_{k=1}^{N-1}\left(2\ln\left(\frac{b}{4}\right)+k(l+2)+\frac{v-\frac{7}{6}}{k}\right)\right]\left(1+O\left(N^{-1}\right)\right)
\end{equation*}
where the error term is uniform with respect to the parameter $0\leq\tau\leq1-\varepsilon<1$. In order to derive \eqref{the20} including the constant term, we now evaluate the sums in the last estimation in the limit $N\rightarrow\infty$ as $\al>1$ and $t>t_0$ (i.e. $t\rightarrow\infty$). This time we use the Euler-Maclaurin type summation formula
\begin{equation}\label{euler2}
	\sum_{k=1}^{N-1}g(k\tau) = \frac{1}{\tau}\int\limits_0^tg(x)\,dx-\frac{1}{2\tau}\int\limits_0^{\tau}g(x)\,dx-\frac{1}{2\tau}\int\limits_{t-\tau}^tg(x)\,dx+R
\end{equation}
with
\begin{equation*}
	R = -\frac{1}{4\tau}\sum_{k=1}^{N-1}\int\limits_{-\tau}^{\tau}\int\limits_0^xg''(k\tau+u)(x-u)\,dudx = O\left(\tau\int\limits_0^t\left|g''(x)\right|\,dx\right)
\end{equation*}
which holds for a twice differentiable function $g:\mathbb{R}\rightarrow\mathbb{R}$. To derive formula \eqref{euler2}, we write the Taylor formula with an integral form for the remainder
\begin{equation*}
	g(s+x) = g(s)+g'(s)x+\frac{1}{2}\int\limits_0^xg''(s+u)(x-u)\,du,
\end{equation*}
then integrate from $-\tau$ to $\tau$,
\begin{equation*}
	\int\limits_{-\tau}^{\tau}g(s+x)\,dx = 2g(s)\tau+\frac{1}{2}\int\limits_{-\tau}^{\tau}\int\limits_0^xg''(s+u)(x-u)\,dudx,
\end{equation*}
and now sum over $\left\{s=k\tau,\ k=1,\ldots,N-1\right\}$,
\begin{equation*}
	\sum_{k=1}^{N-1}\int\limits_{-\tau}^{\tau}g(k\tau+x)\,dx = 2\sum_{k=1}^{N-1}g(k\tau)\tau+\frac{1}{2}\sum_{k=1}^{N-1}\int\limits_{-\tau}^{\tau}\int\limits_0^xg''(k\tau+u)(x-u)\,dudx,
\end{equation*}
which implies \eqref{euler2}.\smallskip

From \eqref{prop2eq1} and \eqref{prop1eq2}, as $N\rightarrow\infty$,
\begin{equation}\label{db2eq1}
	\sum_{k=1}^{N-1}2\ln\left(\frac{b}{4}\right) = -2(N-1)\ln(1-\tau)-\ln t+c_1+O\left(\tau\right)+O\left(t^{-1/2}\right)
\end{equation}
where we introduced as abbreviation
\begin{equation*}
	c_1 = \int\limits_0^1I(8x)\frac{dx}{x}+\int\limits_1^{\infty}\big(I(8x)+1\big)\frac{dx}{x}=-3\ln 2+\int\limits_0^1I(x)\frac{dx}{x}+\int\limits_1^{\infty}\left(I(x)+1\right)\frac{dx}{x} = -\ln 2
\end{equation*}
and the error terms are uniform with respect to the parameters $0\leq\tau\leq1-\varepsilon<1$ and $t>t_0$. Next, combining \eqref{prop2eq1} with \eqref{denseq3}, 
\begin{equation*}
	v = \frac{7}{6}+O\left(\frac{\tau}{\sqrt{1+t}}\right),
\end{equation*}
and hence, 
\begin{equation}\label{db2eq2}
	\sum_{k=1}^{N-1}\frac{v-\frac{7}{6}}{k} = O\left(\tau\right) +O\left(\tau\,t^{-1/2}\right).
\end{equation}
The evaluation of the remaining term involving the Lagrange multiplier will be split into several parts. First
\begin{eqnarray}
	\sum_{k=1}^{N-1}k(l+2) &=& N(N-1)+2(1-\ln 2)(N-1)^2+\sum_{k=1}^{N-1}\ln S(tb)+\sum_{k=1}^{N-1}2J(2bt)\nonumber\\
	&&+\sum_{k=1}^{N-1}k\left(2\ln b-\frac{b}{2}(1-\tau)-b\right)\nonumber\\
	&\equiv&N(N-1)+2(1-\ln 2)(N-1)^2+\Sigma_1+\Sigma_2+\Sigma_3\label{db2eq3}.
\end{eqnarray}
For $\Sigma_1$, use the asymptotic formula
\begin{equation*}
	\ln S(x) = x-\ln(2x)+O\left(e^{-2x}\right),\hspace{0.5cm}x\rightarrow+\infty
\end{equation*} 
and derive
\begin{eqnarray*}
	\Sigma_1 &=&\frac{2\tau}{1-\tau}(N-1)^2-N\ln\left(\frac{8t}{1-\tau}\right)+N+\sqrt{N}\frac{\zeta(3/2)}{\sqrt{2\pi(\al-1)}}+\ln t\\
	&&+\al c_2+c_3+O\left(\tau\right)+O\left(t^{-1/2}\right)+O\left(\al\,e^{-8t}\right),
\end{eqnarray*}
valid in the limit $N\rightarrow\infty,\al>1$ with $t\rightarrow\infty$. Here we have
\begin{eqnarray*}
	c_2 &=&\int\limits_0^{\infty}\left(\ln S(4x)-4x+\ln(8x)\right)\,dx = \frac{1}{8}\int\limits_0^{\infty}\ln\left(1-e^{-x}\right)\,dx=-\frac{\pi^2}{48}\\
	c_3 &=&\frac{3}{2}\ln 2+\int\limits_0^{\infty}4x\left(\frac{S'(4x)}{S(4x)}-1+\frac{1}{4x}\right)\left(1+\frac{I(8x)}{2x}\right)\,dx+2\int\limits_0^{\infty}\left(I(8x)+1-\frac{\zeta(3/2)}{4\sqrt{2\pi x}}\right)\,dx\\
	&&-\frac{1}{2}\int\limits_0^1I(8x)\frac{dx}{x}-\frac{1}{2}\int\limits_1^{\infty}\big(I(8x)+1\big)\frac{dx}{x}=2\ln 2+\frac{\pi^2}{48}+\frac{1}{2}\int\limits_0^{\infty}\frac{I(x)}{e^x-1}\,dx,
\end{eqnarray*}
and we simplified the expressions for $c_i$, by recalling the definitions of $S(x)$ and $I(x)$ as well as the integrals
\begin{equation*}
	\int\limits_0^{\infty}\ln\left(1-e^{-x}\right)\,dx=-\frac{\pi^2}{6},\hspace{0.85cm}\int\limits_0^{\infty}\frac{x\,dx}{e^x-1} = \frac{\pi^2}{6}.
\end{equation*}
Next we go back to Proposition \ref{prop4} and derive
\begin{equation*}
	\Sigma_2 = \al c_4+c_5+O\left(\tau\right)+O\left(\al\,t^{-1/2}\right)
\end{equation*}
where
\begin{eqnarray*}
	c_4 &=&2\int\limits_0^{\infty}J(8x)\,dx=\frac{1}{4}\int\limits_0^{\infty}J(x)\,dx = \frac{\pi^2}{48}\\
	c_5 &=&-(1-\ln 2)+16\int\limits_0^{\infty}xJ'(8x)\left(1+\frac{I(8x)}{2x}\right)\,dx = -(1-\ln 2)-\frac{\pi^2}{48}+\int\limits_0^{\infty}J'(x)I(x)\,dx,
\end{eqnarray*}
in the limit $N\rightarrow\infty,\al>1$, with error terms which are uniform with respect to the parameters $0\leq\tau\leq1-\varepsilon<1$ and $t>t_0$. Here we have used the definite integrals
\begin{equation*}
	\int\limits_0^1\left(\sqrt{\frac{x}{1-x}}-\arctan\sqrt{\frac{x}{1-x}}\right)\frac{dx}{x^2}=\frac{\pi}{2},\hspace{0.85cm}\int\limits_0^1\left(\sqrt{\frac{x}{1-x}}
	-\sqrt{x}\right)\frac{dx}{x^2}=2
\end{equation*}
in order to simplify the expressions for $c_i$. Finally with \eqref{prop2eq1} and Proposition \ref{prop1},
\begin{eqnarray*}
	\Sigma_3 &=&N(N-1)\ln\left(\frac{4}{1-\tau}\right)-N(N-1)-\frac{2}{1-\tau}(N-1)^2-\zeta\left(\frac{3}{2}\right)\al\sqrt{\frac{N}{2\pi(\al-1)}}\\
	&&-\frac{1}{4}\ln t+\al c_6+c_7+O\left(\tau\right)+O\left(\al\,t^{-1/2}\right)
\end{eqnarray*}
where
\begin{eqnarray*}
	c_6 &=&-2\int\limits_0^{\infty}\left(I(8x)+1-\frac{\zeta(3/2)}{4\sqrt{2\pi x}}\right)\,dx=0\\
	c_7 &=&1-\int\limits_0^1I^2(8x)\frac{dx}{4x}-2\int\limits_0^{\infty}\left(I'(8x)+\frac{\zeta(3/2)}{4\sqrt{\pi}(8x)^{3/2}}\right)8x\left(1+\frac{I(8x)}{2x}\right)\,dx\\
	&&+\int\limits_1^{\infty}\left(1-I^2(8x)\right)\frac{dx}{4x}-2\int\limits_0^{\infty}\left(I(8x)+1-\frac{\zeta(3/2)}{4\sqrt{2\pi x}}\right)\,dx\\
	&&-\frac{\zeta(3/2)}{2\sqrt{2\pi}}\int\limits_0^{\infty}\left(4I'(8x)-\frac{I(8x)}{2x}\right)\frac{dx}{\sqrt{x}}=\frac{1}{2}-\frac{3}{4}\ln 2-\frac{1}{4}\int\limits_0^1I^2(x)\frac{dx}{x}+\frac{1}{4}\int\limits_1^{\infty}\left(1-I^2(x)\right)\frac{dx}{x}.
\end{eqnarray*}
Now back to \eqref{db2eq3}, as $N\rightarrow\infty$,
\begin{eqnarray}
	\sum_{k=1}^{N-1}k(l+2)&=&-N^2\ln(1-\tau)+2N\ln(1-\tau)-N\ln(2t)+N-\sqrt{N}\zeta\left(\frac{3}{2}\right)\sqrt{\frac{\al-1}{2\pi}}\nonumber\\
	&&+\frac{3}{4}\ln t+\al(c_2+c_4+c_6)+c_3+c_5+c_7-2\ln 2\nonumber\\
	&&+O\left(\tau\right)+O\left(\al\,t^{-1/2}\right)\label{db2eq4}
\end{eqnarray}
which is uniform with respect to the parameters $0\leq\tau\leq1-\varepsilon<1$ and $t>t_0$. In order to obtain the desired expansion for $\sigma_{N,t}$ we combine estimations \eqref{db2eq1},\eqref{db2eq2} and \eqref{db2eq4}, as $N\rightarrow\infty$
\begin{equation*}
	\sigma_{N,t} = \left(\frac{\al}{\al-1}\right)^{N^2}\left(\frac{\al}{2}\right)^N\hat{C}_{N,t}\,G^{\sqrt{N}}N^{1/4}\left(1+O\left(N^{-1}\right)\right),\hspace{0.25cm} G = \exp\left[-\zeta\left(\frac{3}{2}\right)\sqrt{\frac{\al-1}{2\pi}}\,\right]
\end{equation*}
with
\begin{equation*}
	\hat{C}_{N,t} = (\al-1)^{1/4}\exp\left[d_0\al +c_0+O\left(\tau\right)+O\left(\al\,t^{-1/2}\right)\right]
\end{equation*}
where all error terms are uniform with respect to the parameters $0\leq\tau\leq1-\varepsilon<1$ and $t>t_0$. The factors $d_0$ 
and $c_0$ are given as
\begin{equation}\label{db2eq5}
	d_0= c_2+c_4+c_6 =0
\end{equation}
and
\begin{eqnarray}
	c_0 &=&-\frac{1}{2}-\frac{1}{4}\ln 2+\frac{1}{2}\ln\pi\nonumber\\
	&&+\frac{1}{2}\int\limits_0^{\infty}\frac{I(x)}{e^x-1}\,dx+\int\limits_0^{\infty}J'(x)I(x)\,dx-\frac{1}{4}\int\limits_0^1I^2(x)\frac{dx}{x}
	+\frac{1}{4}\int\limits_1^{\infty}\left(1-I^2(x)\right)\frac{dx}{x}.\label{db2eq6}
\end{eqnarray}
We have therefore shown

\begin{prop}\label{prop7} In the limit $N\rightarrow\infty$,
\begin{equation}\label{prop7eq1}
	Z_N = \hat{C}_{N,t}F^{N^2}G^{\sqrt{N}}N^{1/4}\left(1+O\left(N^{-1}\right)\right),\hspace{0.5cm} F=\frac{\al+1}{2},\ \ G=\exp\left[-\zeta\left(\frac{3}{2}\right)\sqrt{\frac{\al-1}{2\pi}}\right]
\end{equation}
with
\begin{equation*}
	\hat{C}_{N,t} = (\al-1)^{1/4}\exp\left[c_0+O\left(\tau\right)+O\left(\al\,t^{-1/2}\right)\right],
\end{equation*}
where $c_0$ is given explicitly in \eqref{db2eq6} and the error terms are uniform with respect 
to the parameters $0\leq\tau\leq1-\varepsilon<1$ and $t>t_0$.
\end{prop}
The last proposition allows us to prove Theorem \ref{theorem2}. To this end let us choose $\al$ from a compact subset of the set
\begin{equation*}
	\left\{\al\in\mathbb{R}:\ \al>1\right\}.
\end{equation*}
Then, as $N\rightarrow\infty$,
\begin{equation*}
	Z_N = (\al-1)^{1/4}e^{c_0+O(\tau)}F^{N^2}G^{\sqrt{N}}N^{1/4}\left(1+O\left(N^{-1/2}\right)\right)
\end{equation*}
where the error term is uniform on any compact subset of the set \eqref{excset1}. Comparing the last line with \eqref{todaeq3}, we obtain
\begin{equation*}
	d=0,\qquad c=c_0.
\end{equation*}
In order to derive the stated expression for $c$ in Theorem \ref{theorem2}, we will simplify the integrals appearing in \eqref{db2eq6} as follows. For the last two integrals in \eqref{db2eq6}, we replace one of the factors in the products $I^2(x)$ 
with its definition \eqref{the12}. Evaluating the integrals and recalling our computations for $c_1$, we obtain
\begin{eqnarray*}
	-\frac{1}{4}\int\limits_0^1I^2(x)\frac{dx}{x}+\frac{1}{4}\int\limits_1^{\infty}\left(1-I^2(x)\right)\frac{dx}{x}&=&\frac{1}{2}\ln 2-\frac{1}{4\pi^2}\int\limits_0^1\sqrt{\frac{u}{1-u}}\int\limits_0^1\sqrt{\frac{v}{1-v}}\\
	&&\times\left[\int\limits_0^{\infty}\left(\frac{1}{e^{xv}-1}-\frac{1}{xv}\right)\frac{x\,dx}{e^{xu}-1}\right]dudv.
\end{eqnarray*}
Next, we use geometric progression for $\frac{1}{e^z-1}$ and integrate term by term
\begin{equation*}
	\int\limits_0^{\infty}\left(\frac{1}{e^{xv}-1}-\frac{1}{xv}\right)\frac{x\,dx}{e^{xu}-1} = \sum_{n=1}^{\infty}\left[-\frac{1}{nuv}+\sum_{m=1}^{\infty}\frac{1}{(nu+mv)^2}\right],\ \ 0<u,v<1.
\end{equation*}
With the help of the integrals
\begin{equation*}
	\int\limits_0^1\sqrt{\frac{u}{1-u}}\frac{du}{(1+au)^2} = \frac{\pi}{2(1+a)^{3/2}},\hspace{0.4cm}\int\limits_0^1\frac{du}{\sqrt{1-u}\,(a+u)^{3/2}} = \frac{2}{(1+a)\sqrt{a}},\hspace{0.5cm}a\geq 0
\end{equation*}
this implies
\begin{equation}\label{evsumeq1}
	-\frac{1}{4}\int\limits_0^1I^2(x)\frac{dx}{x}+\frac{1}{4}\int\limits_1^{\infty}\left(1-I^2(x)\right)\frac{dx}{x} = \frac{1}{2}\ln 2-\frac{1}{4\pi}\sum_{n=1}^{\infty}\left[-\frac{\pi}{n}+\sum_{m=1}^{\infty}\frac{1}{(m+n)\sqrt{mn}}\right].
\end{equation}
For the second integral in \eqref{db2eq6} we use the identities
\begin{equation*}
	J'(x) = I'(x)-\frac{1}{\pi}\int\limits_0^1\arctan\sqrt{\frac{u}{1-u}}\frac{d}{dx}\left(\frac{x}{e^{xu}-1}\right)\,du,\hspace{0.75cm} \frac{\partial}{\partial x}\left(\frac{x}{e^{xu}-1}\right) =\frac{\partial}{\partial u}\left(\frac{u}{e^{xu}-1}\right)
\end{equation*} 
and integrate by parts
\begin{equation*}
	\int\limits_0^{\infty}J'(x)I(x)\,dx = \frac{1}{2}-\frac{1}{2}\int\limits_0^{\infty}\frac{I(x)}{e^x-1}\,dx+\frac{1}{2\pi}\int\limits_0^{\infty}\int\limits_0^1\sqrt{\frac{u}{1-u}}\frac{I(x)}{e^{xu}-1}\,dudx.
\end{equation*}
Since
\begin{equation*}
	I(x) = \frac{x}{\pi}\int\limits_0^1\sqrt{\frac{v}{1-v}}\left(\frac{1}{e^{xv}-1}-\frac{1}{xv}\right)\,dv,
\end{equation*}
we obtain
\begin{equation*}
	\frac{1}{2\pi}\int\limits_0^{\infty}\int\limits_0^1\sqrt{\frac{u}{1-u}}\frac{I(x)}{e^{xu}-1}\,dudx = \frac{1}{2\pi^2}\int\limits_0^1\sqrt{\frac{u}{1-u}}\int\limits_0^1\sqrt{\frac{v}{1-v}}\left[\int\limits_0^{\infty}\left(\frac{1}{e^{xv}-1}-\frac{1}{xv}\right)\frac{x\,dx}{e^{xu}-1}\right]dudv,
\end{equation*}
i.e. the triple integral we just evaluated in the computation of \eqref{evsumeq1}. We summarize
\begin{equation}\label{evsumeq2}
	\int\limits_0^{\infty}J'(x)I(x)\,dx = \frac{1}{2}-\frac{1}{2}\int\limits_0^{\infty}\frac{I(x)}{e^x-1}\,dx+\frac{1}{2\pi}\sum_{n=1}^{\infty}\left[-\frac{\pi}{n}+\sum_{m=1}^{\infty}\frac{1}{(m+n)\sqrt{mn}}\right]
\end{equation}
and back to \eqref{db2eq6},
\begin{equation*}
	c = \frac{1}{4}\ln 2+\frac{1}{2}\ln\pi+\frac{1}{4\pi}\sum_{n=1}^{\infty}\left[-\frac{\pi}{n}+\sum_{m=1}^{\infty}\frac{1}{(m+n)\sqrt{mn}}\right],
\end{equation*}
thus proving Theorem \ref{theorem2}.



\begin{thebibliography}{100}

\bibitem{AR}
D. Allison and N. Reshetikhin, Numerical study of the $6$-vertex model with domain wall boundary conditions, {\it Ann. Inst. Fourier} (Grenoble) {\bf 55} (2005), 1847-1869.

\bibitem{BBdF}
J. Baik, R. Buckingham and J. DiFranco, Asymptotics of Tracy-Widom distributions and the total integral of a Painlev\'e II function, {\it Commun. Math. Phys.} {\bf 280} (2008), 463-497

\bibitem{BT}
E. Basor and C. Tracy, Some problems associated with the asymptotics of $\tau$-functions, {\it Surikagaku (Mathematical Sciences)} {\bf 30}, no.\,3 (1992), 71-76.

\bibitem{BE}
H. Bateman, A. Erdelyi, {\it Higher transcendental functions}, McGraw-Hill, NY, 1953.

\bibitem{BB}
P. Bleher and T. Bothner, Exact solution of the six-vertex model with domain wall-boundary conditions. Critical line between disordered and antiferroelectric phases, {\it Random Matrices: Theory Appl.} {\bf 01}, (2012) 1250012.

\bibitem{BF}
P. Bleher and V. Fokin, Exact solution of the six-vertex model with domain wall-boundary conditions. Disordered phase, {\it Commun. Math. Phys.} {\bf 268} (2006), 223-284. 

\bibitem{BL1}
P. Bleher and K. Liechty, Exact solution of the six-vertex model with domain wall-boundary conditions. Ferroelectric phase, {\it Commun. Math. Phys.} {\bf 286} (2009), 777-801.

\bibitem{BL2}
P. Bleher and K. Liechty, Exact solution of the six-vertex model with domain wall-boundary conditions. Critical line between ferroelectric and disordered phases, {\it J. Statist. Phys.} {\bf 134} (2009), 463-485.

\bibitem{BL3}
P. Bleher and K. Liechty, Exact solution of the six-vertex model with domain wall-boundary conditions. Antiferroelectric phase, {\it Commun. Pure Appl. Math.} {\bf 63} (2010), 779-829.

\bibitem{BK} 
P. Bleher and A. Kuijlaars, Large $n$ limit of Gaussian random matrices with external source, part III: Double scaling limit, {\it Commun. Math. Phys.} {\bf 270} (2007), 481-517.

\bibitem{BI}
T. Bothner and A. Its, Asymptotics of a Fredholm determinant corresponding to the first bulk critical universality class in random matrix models, {\it Commun. Math. Phys.}  {\bf 328} (2014), 155-202.

\bibitem{BuB}
A. Budylin, V. Buslaev, Quasiclassical asymptotics of the resolvent of an integral convolution operator with a sine kernel on a finite interval, (Russian) {\it Algebra i Analiz} {\bf 7}, no.\,6 (1995), 79-103

\bibitem{CIK}
D. A. Coker, A. G. Izergin and V. E. Korepin, Determinant formula for the six-vertex model, {\it J. Phys. A} {\bf 25}, (1992) 4315-4334.

\bibitem{DIKZ}
P. Deift, A. Its, I. Krasovsky and X. Zhou, The Widom-Dyson constant for the gap probability in random matrix theory, {\it J. Comput. Appl. Math.} {\bf 202} no.\,1 (2007), 26-47

\bibitem{DIK}
P. Deift, A. Its and I. Krasovsky, Asymptotics of the Airy-kernel determinant, {\it Commun. Math. Phys.} {\bf 278} no.\,3 (2008), 643-678

\bibitem{DKM}
P. A. Deift, T. Kriecherbauer and K. T-R. McLaughlin, New results on equlibirum measure for logarithmic potentials in the presence of an external field, {\it J. Approx. Theory} {\bf 95} (1998), 388-475.

\bibitem{DKMVZ}
P. A. Deift, T. Kriecherbauer, K. T-R. McLaughlin, S. Venakides and X. Zhou, Uniform asymptotics for polynomials orthogonal with respect to varying exponential weights and applications to universality questions in random matrix theory, {\it Commun. Pure Appl. Math.} {\bf 52} (1999), 1335-1425.

\bibitem{DZ}
P. A. Deift and X. Zhou, A steepest descent method for oscillatory Riemann-Hilbert problems. Asymptotics for the MKdV equation, {\it Ann. of Math.}, {\bf 137} (1993), 295-368.

\bibitem{E}
T. Ehrhardt, Dyson's constant in the asymptotics of the Fredholm determinant of the sine kernel, {\it Commun. Math. Phys.} {\bf 262} (2006), 317-341

\bibitem{FS}
P. L. Ferrari and H. Spohn, Domino tilings and the six-vertex model at its free fermion point, {\it J. Phys. A: Math. Gen.} {\bf 39} (2006) 1029710306.

\bibitem{FIK}
A. S. Fokas, A. R. Its and A. V. Kitaev, Discrete Painlev\'e equations and their appearance in quantum gravity, {\it Comm. Math. Phys.} {\bf 142} (2) (1991), 313-344.

\bibitem{F}
P. Forrester, Asymptotics of spacing distributions $50$ years later, to appear in Proceedings of the MSRI semester ``Random matrix theory, interacting particle systems and integrable systems", preprint: arXiv: 1204.3225v3

\bibitem{I}
A. G. Izergin, Partition function of the six-vertex model in a finite volume. {\it Dokl. Akad. Nauk SSSR} {\bf 297}, no. 2, (1987), 331-333; translation in Soviet Phys. Dokl. {\bf 32} (1987), 878-880.

\bibitem{K}
V. Korepin, Calculation of norms of Bethe wave functions, {\it Commun. Math. Phys.} {\bf 86} (1982), 391-418.

\bibitem{KZ}
V. Korepin and P. Zinn-Justin, Thermodynamic limit of the six-vertex model with domain wall boundary conditions, {\it J. Phys. A} {\bf 33}(40), (2000) 7053.

\bibitem{T}
C. Tracy, Asymptotics of the $\tau$-function in the two-dimensional Ising model, {\it Commun. Math. Phys.} {\bf 142} (1991), 297-311.

\bibitem{V}

M. Vanlessen, Strong asymptotics of Laguerre-type orthogonal polynomails and applications in random matrix theory, {\it Construct. Approx.} {\bf 25} (2007), 125-175.

\bibitem{Z}
P. Zinn-Justin, Six-vertex model with domain wall boundary conditions and one-matrix model, {\it Phys. Rev. E} {\bf 62}, (2000) 3411-3418.

\end{thebibliography}
\end{document}